\documentclass[11pt]{article}
\usepackage{amssymb}
\usepackage{amsthm}
\usepackage{latexsym}
\usepackage{amsmath}
\usepackage{graphicx}
\newtheorem{thm}{Theorem}[section]
\newtheorem{prop}[thm]{Proposition}
\newtheorem{cor}[thm]{Corollary}
\newtheorem{lem}[thm]{Lemma}
\theoremstyle{definition}
\newtheorem{defi}[thm]{Definition}
\newtheorem{exa}[thm]{Example}
\newtheorem{rem}[thm]{Remark}
\numberwithin{equation}{section}

\newcommand{\TS}{\textstyle}
\newcommand{\N}{{\mathbb N}}
\newcommand{\F}{{\mathbb F}}

\newcommand{\Z}{{\mathbb Z}}

\newcommand{\C}{{\mathbb C}}

\newcommand{\cC}{{\mathcal C}}

\newcommand{\cM}{{\mathcal M}}
\newcommand{\cP}{{\mathcal P}}

\newcommand{\cQ}{{\mathcal Q}}

\newcommand{\cR}{{\mathcal R}}

\newcommand{\cU}{{\mathcal U}}
\newcommand{\cV}{{\mathcal V}}

\newcommand{\wA}{{\widehat{A}}}
\newcommand{\wR}{{\widehat{R}}}

\newcommand{\wM}{{\widehat{M}}}
\newcommand{\lR}[1]{{\mbox{$_R{#1}$}}}
\newcommand{\rR}[1]{{\mbox{${#1}_R$}}}
\newcommand{\lp}[1]{{\mbox{$^{\perp}{#1}$}}}
\newcommand{\rp}[1]{{\mbox{${#1}^{\perp}$}}}
\newcommand{\biMod}[3]{{\mbox{$_{#1}{#2}\!{\,}_{#3}$}}}

\newcommand{\bidual}[1]{{\widehat{\phantom{\big|}\hspace*{.5em}}\hspace*{-.9em}\widehat{#1}}}
\newcommand{\wcP}{\widehat{\mathcal P}}
\newcommand{\cPHam}{\mbox{${\mathcal P}_{\rm Ham}$}}
\newcommand{\wcQ}{\widehat{\mathcal Q}}
\newcommand{\wwcP}{\widehat{\phantom{\big|}\hspace*{.5em}}\hspace*{-.9em}\wcP}
\newcommand{\wcPchil}{\wcP^{^{\scriptscriptstyle[\chi,l]}}}
\newcommand{\wcPchir}{\wcP^{^{\scriptscriptstyle[\chi,r]}}}
\newcommand{\wcQchil}{\wcQ^{^{\scriptscriptstyle[\chi,l]}}}
\newcommand{\wcQchir}{\wcQ^{^{\scriptscriptstyle[\chi,r]}}}
\newcommand{\widesim}[1][1.5]{\scalebox{#1}[1]{$\sim$}}
\newcommand{\sbt}{\raisebox{.2ex}{\mbox{$\scriptscriptstyle\bullet\,$}}}
\newcommand{\betal}{\mbox{$\beta_{\rm l}$}}
\newcommand{\betar}{\mbox{$\beta_{\rm r}$}}
\newcommand{\alphal}{\mbox{$\alpha_{\rm l}$}}
\newcommand{\alphar}{\mbox{$\alpha_{\rm r}$}}
\newcommand{\Hom}{\mbox{\rm Hom}}
\newcommand{\End}{\mbox{\rm End}}
\newcommand{\Aut}{\mbox{\rm Aut}}
\newcommand{\LT}{\mbox{\rm LT}}
\newcommand{\Mon}{\mbox{\rm Mon}}
\newcommand{\soc}{\mbox{\rm soc}}
\newcommand{\rad}{\mbox{\rm rad}}

\newcommand{\spann}{\mbox{\rm span}\,}
\newcommand{\swc}{\mbox{\rm swc}}

\newcommand{\wt}{{\rm wt}}
\newcommand{\wtH}{\mbox{{\rm wt}$_{\rm H}$}}
\newcommand{\wtN}{\mbox{{\rm wt}$_{\rm N}$}}
\newcommand{\wtRT}{\mbox{{\rm wt}$_{\rm RT}$}}
\newcommand{\wtRTr}{\mbox{{\rm wt}$_{{\rm RT},r}$}}

\newcommand{\T}{\mbox{$\!^{\sf T}$}}
\newcommand{\ideal}[1]{\mbox{$\langle{#1}\rangle$}}
\newcommand{\GL}{\mathrm{GL}}
\newcommand{\supp}{\mathrm{supp}}
\newcommand{\ev}{\mathrm{ev}}

\usepackage{hyperref}
\usepackage{accents}
\newlength{\dhatheight}

\newcommand{\Smallfourmat}[4]{\mbox{$\left(\begin{smallmatrix}{#1}&{#2}\\{#3}&{#4}\end{smallmatrix}\right)$}}

\newcounter{alp}
\newcounter{ara}
\newcounter{rom}
\newenvironment{romanlist}{\begin{list}{(\roman{rom})\hfill}{\usecounter{rom}
     \topsep0ex \labelwidth.7cm \leftmargin.7cm \labelsep0cm
     \rightmargin0cm \parsep0ex \itemsep.4ex
     \partopsep1ex}}{\end{list}}
\newenvironment{alphalist}{\begin{list}{(\alph{alp})\hfill}{\usecounter{alp}
     \topsep0ex \labelwidth.6cm \leftmargin.6cm \labelsep0cm
     \rightmargin0cm \parsep0ex \itemsep0ex
     \partopsep0ex}}{\end{list}}
\newenvironment{arabiclist}{\begin{list}{(\arabic{ara})\hfill}{\usecounter{ara}
     \topsep0ex \labelwidth.6cm \leftmargin.6cm \labelsep0cm
     \rightmargin0cm \parsep0ex \itemsep0ex
     \partopsep1.6ex}}{\end{list}}

\title{Extension Theorems for Various Weight Functions\\ over Frobenius Bimodules}
\date{\today}
\author{Heide Gluesing-Luerssen$^*$
and Tefjol Pllaha\footnote{HGL was partially supported by the National
Science Foundation Grant DMS-1210061 and by the grant \#422479 from the Simons Foundation.
HGL and TP are with the Department of Mathematics, University of Kentucky, Lexington KY 40506-0027, USA;
\{heide.gl, tefjol.pllaha\}@uky.edu.}}

\begin{document}
\maketitle

{\bf Abstract:}
In this paper we study codes where the alphabet is a finite Frobenius bimodule over a finite ring.
We discuss the extension property for various weight functions.
Employing an entirely character-theoretic approach and a duality theory for partitions on Frobenius bimodules we derive
alternative proofs for the facts that the Hamming weight and the homogeneous weight satisfy the extension property.
We also use the same techniques to derive the extension property for other weights, such as the Rosenbloom-Tsfasman weight.

{\bf Keywords:} Frobenius bimodules, codes, weight functions, extension property.

{\bf MSC (2010):} 94B05, 16L60, 16P10

\section{Introduction}
We discuss the extension property in the sense of MacWilliams for weight-preserving maps.
The extension property can be described most generally as follows:
given a substructure (a code)~$\cC$ of some~$M^n$, where~$M$ is a finite ring or module, and a
homomorphism $f:\cC\longrightarrow M^n$ that preserves a certain weight function, such as the Hamming weight or some other property, does
this map extend to a map on the entire space~$M^n$ that still preserves the given weight or property?
For the Hamming weight this question has been settled in the affirmative for fields by MacWilliams~\cite{MacW62},
for finite Frobenius rings by Wood~\cite{Wo97,Wo99} and for finite Frobenius bimodules over a finite ring by
Greferath et al.~\cite{GNW04}.
In each of these cases the Hamming weight-preserving map is given by a monomial transformation.
In addition, some other weight functions have been investigated in the above-mentioned literature as well and according extension
results have been established.
In~\cite{BGL15} Barra et al.\ consider finite Frobenius rings and use a different approach, called the local-global principle,
in order to provide further cases where maps with a particular pointwise (local) property can be extended while preserving said property.
A main tool is the notion of partitions of the ambient space and their character-theoretic dual.

In this paper we focus on finite Frobenius bimodules,~$M$, over a finite ring,~$R$, and lend some tools from~\cite{BGL15} in order to
discuss the extension property for various weight functions.
This allows us to consolidate the above-mentioned results into a unified framework that is solely based on character theory.
In addition, we discuss some other weight functions  that satisfy the extension property.
Our main tools are partitions on~$M^n$ and their character-theoretic duals, which are partitions on~$R^n$.
Special emphasis is put on orbit partitions induced by group actions on~$M^n$.
Their dual partitions turn out to be the orbit partitions on~$R^n$ of the `transposed group action'.
We will then use this duality in two different ways.
Firstly, for certain weights the weight-preserving maps give rise to pointwise expressions in terms of an associated matrix group.
If the group is nicely structured, character sums allow us to turn this into a uniform matrix from that group representing the map, thus
establishing the extendability of the map.
For other types of weights, such as the Hamming weight, this approach does not work as the associated matrix group, i.e., monomial matrices, is not suitably structured. Instead we make use of the fact that any weight-preserving map preserves the associated weight partition and,
again with the aid of character sums and the duality theory for partitions, we obtain the desired uniform matrix.

The paper is organized as follows.
In Section~\ref{S-basics} we derive basic properties of finite Frobenius bimodules over finite
rings solely based on a character-theoretic approach.
In particular, we derive the well-known double annihilator properties.
Section~\ref{S-Partitions} is devoted to partitions on~$M^n$, where~$M$ is a Frobenius bimodule over a finite ring.
We introduce a character-theoretic dualization of such partitions.
A special case is given by orbit partitions induced by group actions.
In Section~\ref{S-Extensions} we turn to the extension property for various weight functions.
Using the dualization technique for partitions, we derive the well-known extension property for the Hamming weight and the homogeneous weight.
With the same technique we also establish the extension property for the Rosenbloom-Tsfasman weight and some other weight-like functions.
In Remark~\ref{R-OtherWeights} we summarize how the techniques may be used to discuss the extension property for other weights
such as the big class of poset weights.
With the methods developed thus far, this turns out to be entirely analogous to~\cite{BGL15} where the same question has been studied
for codes over finite Frobenius rings.
Finally, in the last section we establish the extension property for the Rosenbloom-Tsfasman weight applied to $\F$-linear maps of
$\F$-subspaces of some $\hat{\F}^n$, where $\hat{\F}$ is a field extension of~$\F$.
Even though $\hat{\F}$ is not a Frobenius bimodule over~$\F$, this case  can be dealt with using the methods developed in this paper.

\section{Finite Frobenius Rings and Bimodules}\label{S-basics}
In this section we derive the basic properties of finite Frobenius bimodules over finite rings.
These properties are well-known (see, for instance,~\cite{GNW04} by Greferath et al.\ and~\cite{Wo09} by Wood)
but we will give a different, purely character-theoretic approach and establish the results without resorting to quasi-Frobenius
bimodules.

We begin with the character group of a finite abelian group.
Let~$A$ be a finite abelian group.
Its character group is defined as the set $\widehat{A} := \text{Hom}(A, \mathbb{C}^*)$ of all group homomorphisms from~$A$ to~$\C^*$
endowed with addition $(\chi_1 + \chi_2)(a) = \chi_1(a)\chi_2(a)$ for all $\chi_i\in\widehat{A}$ and $a\in A$.
Then $\widehat{A}$ is an abelian group.
Its zero element is $\varepsilon \in \widehat{A}$ given by $\varepsilon(a) =1$ for all $a\in A$.
Elements of $\widehat{A}$ are called \textit{characters} and $\varepsilon$ is the \textit{principal character} of $A$.
We list some basic properties of the character group. They are well-known and/or can easily be verified.

\begin{rem}\label{R-CharGroup}
Let~$A$ be a finite abelian group.
\begin{arabiclist}
\item $A$ is isomorphic to~$\widehat{A}$ (though not naturally so) and hence $|A| = |\widehat{A}|$.
\item $\widehat{A_1}\times \widehat{A_2}\cong\widehat{A_1\times A_2}$ for any finite abelian
      groups~$A_1$ and~$A_2$. The isomorphism is facilitated by $(\chi_1,\chi_2)(a_1,a_2):=\chi_1(a_1)\chi_2(a_2)$.
\item $A$ and $\bidual{A}$ are naturally isomorphic via the map $\zeta_A: a\longmapsto\ev_a$, where $\ev_a: \widehat{A}\longrightarrow\C^*$,
      $\chi\longmapsto\chi(a)$ denotes the evaluation map.
\item  $\sum_{a\in A}\varepsilon(a)=|A|$ and $\sum_{a\in A}\chi(a)=0$ for $\chi\in\widehat{A}\setminus\{\varepsilon\}$.
\item Distinct characters of~$A$ are linearly independent in the $\C$-vector space of maps from~$A$ to~$\C$.
\item Let $\chi_1, \ldots, \chi_N$ and $\chi'_1, \ldots , \chi'_M$ be characters of $A$.
     If $\sum_{i=1}^N \chi_i = \sum_{i=1}^M \chi'_i$
     as maps from~$A$ to~$\C$, then the multisets $\{\!\{\chi_1, \ldots, \chi_N\}\!\}$ and
     $\{\!\{\chi'_1, \ldots , \chi'_M\}\!\}$ coincide, see \cite[Prop.~3.1]{BGL15}.
\item Let $B\leq A$ and $C\leq\wA$ be subgroups of~$A$ and~$\wA$, respectively.
      Their \emph{dual groups} are defined as
      $B^\circ:=\{\chi\in\wA\mid B\subseteq\ker\chi\}$ and
      $C^\circ:=\{a\in A\mid a\in\ker\chi\text{ for all }\chi\in C\}$,
      where for $\chi\in\widehat{A}$ we set $\ker\chi=\{x\in A\mid \chi(x)=1\}$.
      Clearly,~$B^\circ$ and~$C^\circ$ are subgroups of~$\wA$ and~$A$, respectively.
      Then
      \begin{romanlist}
      \item $B^\circ\cong \widehat{A/B}$, thus $|B^\circ|=|A|/|B|$.
      \item If $B\subseteq\ker\chi$ for all $\chi\in\wA$, then $B=\{0\}$.
      \item $(B^\circ)^\circ=B$ and $(C^\circ)^\circ=C$.
      \end{romanlist}
\end{arabiclist}
\end{rem}

From now on let~$R$ be a finite ring with identity. Consider the character group
$\widehat{R}$ of the abelian group $(R,+)$.
Then $\wR$ is an $(R,R)$-bimodule with the left and right scalar multiplication
\begin{equation}\label{e-bimodstruc}
    (r\!\cdot\!\chi)(a) = \chi(ar)\ \text{  and }\ (\chi\!\cdot\!r)(a) = \chi(ra) \text{ for all }\chi\in\wR\text{ and }a,r\in R.
\end{equation}
Remark~\ref{R-CharGroup}(3) yields $\bidual{R}\cong R$ as $(R,R)$-bimodules.
Furthermore, the dual groups from Remark~\ref{R-CharGroup}(7) of left (resp.\ right) submodules are
right (resp.\ left) submodules.

\begin{prop}\label{P-rwR}
Let~$R$ be any finite ring with identity and $r\in R$ such that $r\wR =\{\varepsilon\}$ or $\wR r=\{\varepsilon\}$. Then $r=0$.
In other words, $\lR{\wR}$ and $\rR{\wR}$ are faithful.
\end{prop}
\begin{proof}
Without loss of generality assume $r\wR=\{\varepsilon\}$.
Then $1=(r\chi )(s)=\chi(sr)$ for all $\chi\in\wR$ and $s\in R$.
Thus $Rr\subseteq\ker\chi$ for all $\chi\in\wR$.
Remark~\ref{R-CharGroup}(7)(ii) concludes the proof.
\end{proof}

The isomorphism in Proposition~\ref{P-RCongEndRhat} below will be crucial for dealing with Frobenius bimodules.
It tells us that all left-linear maps on~$\wR$ are given by right multiplication with a ring element.
For this reason we will identify the rings~$R$ and $\End(\lR{\wR})$.
The result appears also in \cite[Thm.~2.4]{GNW04}.
For the sake self-containedness we provide an elementary proof.
First a convention.
\begin{rem}\label{R-EndM}
Let $\lR{M}$ be a left $R$-module.
We will adopt the convention that multiplication in the endomorphism ring $S:=\End(\lR{M})$
is defined as $f\!\cdot\!g:=g\circ f$, where the latter means $(g\circ f)(v)=g(f(v))$ for $v\in M$.
We will not adhere to the convention of using right operators for left linear maps because of possible confusion when composing
left and right linear maps.
Thus, composition of maps will always mean $(g\circ f)(v)=g(f(v))$,
and multiplication in the endomorphism ring reverses the order.
This turns $M$ into a right $S$-module via $v\cdot\!f:=f(v)$.
\end{rem}

\begin{prop}\label{P-RCongEndRhat}
For any finite ring~$R$ with identity the map
\[
   \Theta:R\longrightarrow \End(\lR{\wR}),\ r\longmapsto
   \left\{\begin{array}{rcl}\Theta(r):\lR{\wR}&\longrightarrow&\lR{\wR}\\ \chi&\longmapsto&\chi r\end{array}\right\}
\]
is an isomorphism of rings and of right $R$-modules.
By symmetry, $R\cong\End(\rR{\wR})$ as rings and left $R$-modules.
As a consequence, the bimodule $\biMod{R}{\wR}{R}$ is balanced in the sense of~\cite[Sec.~1]{GNW04}.
\end{prop}
\begin{proof}
First of all, it is clear that the map $\chi\mapsto\chi r$ is indeed in $\End(\lR{\wR})$.
One also easily verifies that~$\Theta$ is a ring homomorphism as well as right $R$-linear.
Injectivity of~$\Theta$ is immediate from Proposition~\ref{P-rwR}.
For surjectivity we make use of the evaluation character on~$\wR$; see Remark~\ref{R-CharGroup}(3).
Note that for any $r\in R$ we have $\ev_1\circ\Theta(r)=\ev_r$.
Thus Remark~\ref{R-CharGroup}(3) implies $\bidual{R}=\{\ev_1\circ\Theta(r)\mid r\in R\}$.
Let now $f\in \End(\lR{\wR})$.
Then $\ev_1\circ f\in\bidual{R}$, and hence $\ev_1\circ f=\ev_1\circ\Theta(r)$ for some $r\in R$.
This means $\big(f(\chi)\big)(1)=(\chi r)(1)$ for all $\chi\in\wR$.
Applying this to the characters $s\chi$ for any $s\in R$ and using the left linearity of~$f$ and the left module structure
of~$\wR$ we obtain $\big(f(\chi)\big)(s)=(\chi r)(s)$ for all $s\in R$.
Therefore $f(\chi)=\chi r$, which in turn yields $f=\Theta(r)$, as desired.
\end{proof}

Recall that~$R$ is a Frobenius ring if $_R\soc(_R R)\cong\, _R(R/\rad(R))$,
where $\soc(_R R)$ denotes the socle of the left~$R$-module~$R$
and $\rad(R)$ is the Jacobson radical of~$R$.
Honold~\cite[Thm.~2]{Hon01} showed that the existence of a left isomorphism implies a right analogue.
It follows from Lamprecht~\cite{Lamp53} (see also Hirano~\cite[Th.~1]{Hi97}, Honold~\cite[p.~409]{Hon01}, and
Wood~\cite[Th.~3.10]{Wo99}) that the Frobenius property of finite rings can also be characterized via the character module.
Since this is all we need in this paper, we will use this as our definition of finite Frobenius rings.

\begin{defi}\label{D-FrobRing}
A finite ring $R$ with identity is called \emph{Frobenius} if $\lR{R}\cong\lR{\wR}$, that is,
if there exists $\chi \in \widehat{R}$ such that the map
$r \mapsto r\chi$ is an isomorphism of left $R$-modules.
In this case~$\chi$ is called a \emph{generating character of~$R$}.
Stated differently,~$R$ is Frobenius with generating character~$\chi$ iff $_R\wR$ is a free $R$-module with basis $\{\chi\}$.
\end{defi}
The same references as above also show that $\{\chi\}$ is a basis of $\lR{\wR}$ if and only if $\{\chi\}$ is a basis of $\rR{\wR}$.
As a consequence, $\lR{R}\cong \lR{\wR}\Longleftrightarrow \rR{R}\cong\rR{\wR}$ and there is no need of specifying the sidedness in the
definition of Frobenius and generating characters.

Examples of finite Frobenius rings are finite fields, integer residue rings $\mathbb{Z}_N:=\Z/N\Z$, finite chain rings as well as
matrix rings $R^{n\times n}$ and finite group rings $R[G]$ over Frobenius rings~$R$.
Direct products of finite Frobenius rings are Frobenius.
A simple example of a non-Frobenius ring is the (commutative) ring $R=\F_2[x,y]/(x^2,y^2,xy)$, see~\cite[Ex.~3.2]{ClGo92}.

We now turn to modules.
Let~$M$ be a finite left $R$-module.
Its underlying abelian group $(M,+)$ gives rise to the character group $\wM$,
which is endowed with a right $R$-module structure via $(\chi\!\cdot\!r)(v) = \chi(rv)$ for all $\chi\in\wM,\,r\in R,\,v\in M$.
Similarly, for a right $R$-module~$M$ the character group~$\wM$ carries a left $R$-module structure via $(r\!\cdot\!\chi)(v)=\chi(vr)$.
As for rings, we have for any module~$\lR{M}$ the canonical isomorphism
\begin{equation}\label{e-bidualM}
  \zeta_M:\lR{M}\longmapsto\lR{\widehat{\phantom{\big|}\hspace*{.8em}}\hspace*{-1.2em}\widehat{M}}\ \text{ via }\
   v\longmapsto \left\{\begin{array}{rcl}\ev_v:\wM&\longrightarrow&\C^*\\ \chi&\longmapsto&\chi(v)\end{array}\right.
\end{equation}
and analogously for right modules.

The following result, going back to Bass' Theorem, will be crucial at several instances.
\begin{thm}[\mbox{\cite[Prop.~5.1]{Wo99}}]\label{T-LeftSim}
Let~$R$ be any finite ring with identity and~$M$ a finite left $R$-module. Let $v,w\in M$ be such that $Rw=Rv$.
Then there exists a unit $\alpha\in R^*$ such that $w=\alpha v$.
\end{thm}

Let now $M$ be an $(R,R)$-bimodule.
As with Frobenius rings, the following definition of Frobenius bimodules differs from the usual one in the literature,
see for instance~\cite{GNW04}.
It does, however, also appear in \cite[Sec.~5.2]{Wo09}.
In particular, the definition below will not explicitly resort to quasi-Frobenius bimodules.
We will see later that the following definition is indeed equivalent to the one given in \cite[Def.~2.16]{GNW04}.

\begin{defi}\label{D-FrobBimod}
Let~$R$ be any finite ring with identity and~$M$ a finite $(R,R)$-bimodule.
Then~$M$ is called a \emph{Frobenius bimodule} if $\lR{M} \cong\lR{\wR}$ and $\rR{M} \cong \rR{\wR}$.
\end{defi}
Clearly, a finite ring~$R$ is a Frobenius $(R,R)$-bimodule if and only if it is a Frobenius ring.
Moreover,~$\biMod{R}{\wR}{R}$ is a Frobenius bimodule for any finite ring~$R$.
Recalling that $\biMod{R}{\bidual{R}}{R}\cong\biMod{R}{R}{R}$ for any finite ring~$R$ we conclude that if~$R$ is not a
Frobenius ring then $\biMod{R}{M}{\!R}$ being Frobenius does not imply that
$\biMod{R}{\wM}{\!R}$ is Frobenius.
In general the condition of~$M$ being a Frobenius bimodule is quite restrictive because it implies that $|M|=|\wR|=|R|$.
In particular, a vector space over a finite field~$\F$ is a Frobenius bimodule over~$\F$ iff it is one-dimensional.
One should note that a Frobenius bimodule need not be isomorphic to~$\wR$ as an $(R,R)$-bimodule.

In order to discuss some properties of the structure of Frobenius bimodules we start with the following consideration.

\begin{rem}\label{R-GenCharMod}
Let~$R$ be a finite ring with identity and~$\lR{M}$ be a finite  left $R$-module.
Suppose there exists an isomorphism $\lambda:\lR{M}\longrightarrow\lR{\wR}$.
Then one easily verifies that the induced map
\[
   \rR{R}\longrightarrow\rR{\wM},\ r\longmapsto\zeta_R(r)\circ\lambda=\ev_r\circ\lambda, \text{ where }
   \zeta_R:\rR{R}\longrightarrow\bidual{R}{\,}_{\!R}\text{ is as in~\eqref{e-bidualM}},
\]
is an isomorphism of right $R$-modules.
As a consequence, the module $\rR{\wM}$ is free of rank~$1$.
Any basis vector of $\rR{\wM}$ is called a \emph{right generating character} of~$M$.
The above shows that the right generating characters of~$M$ are given by $\zeta_R(u)\circ\lambda=\ev_u\circ\lambda$ where $u\in R^*$.
Analogous results are true for right $R$-modules.
Thus, all of the above shows that a finite Frobenius bimodule $\biMod{R}{M}{\!R}$ has a left generating character and a right
generating character.
It is a consequence of Theorem~\ref{T-LeftSim} (and can also easily be verified using the above) that if~$\chi$ and~$\chi'$ are right (resp.\ left)
generating characters of~$M$, then there exists a unit $u\in R^*$ such that
$\chi'=\chi u$ (resp.\ $\chi'=u\chi$).
\end{rem}

In the following theorem we list some basic facts about generating characters of Frobenius bimodules.
Recall the kernel of a character from Remark~\ref{R-CharGroup}(7).
The equivalences can be found in \cite[Lem.~5.2, Lem.~5.3, Cor.~5.1]{Wo09}.
The last part follows easily from~(iii).

\begin{thm}\label{T-FrobBimod}
Let~$\biMod{R}{M}{\!R}$ be a Frobenius bimodule and let~$\chi\in\wM$.
The following are equivalent.
\begin{romanlist}
\item $\chi$ is a left generating character of~$M$, i.e., $\lR{\wM}=R\chi$.
\item $\chi$ is a right generating character of~$M$, i.e., $\rR{\wM}=\chi R$.
\item $\ker\chi$ contains no nonzero left submodule of~$M$.
\item $\ker\chi$ contains no nonzero right submodule of~$M$.
\end{romanlist}
Furthermore, if $\chi$ is a generating character of~$M$ and~$V$ is any left $R$-module then the map
$\Hom(\lR{V},\lR{M}) \longrightarrow\widehat{V},\ g \longmapsto \chi\circ g$ is an injective group homomorphism.
\end{thm}

Now we can fix particular isomorphisms.
From now one let $M=\biMod{R}{M}{\!R}$ be a finite Frobenius bimodule with generating character~$\chi$ over the finite ring~$R$.
Then one easily verifies that the maps
\begin{equation}\label{e-beta}
  \betal:\lR{M}\longrightarrow \lR{\wR},\ v\longmapsto \chi(\sbt v) \ \text{ and }\
  \betar:\rR{M}\longrightarrow \rR{\wR},\ v\longmapsto \chi(v \sbt)
\end{equation}
as well as
\begin{equation}\label{e-alpha}
  \alphal: \lR{R}\longrightarrow \lR{\wM},\ r\longmapsto \chi(\sbt r)=r\chi\ \text{ and }\
  \alphar: \rR{R}\longrightarrow \rR{\wM},\ r\longmapsto \chi(r \sbt)=\chi r
\end{equation}
are isomorphisms (as common, $\sbt$ stands for the argument of the character in question).
In fact, these pairs of isomorphisms mutually induce each other in the sense of Remark~\ref{R-GenCharMod}.
Precisely,
\begin{align}
  &\alphar(r)=\zeta_R(r)\circ \betal, \qquad
  \alphal(r)=\zeta_R(r)\circ\betar & &\hspace*{-3em}\text{ for all }r\in R,\label{e-alphabeta}\\[.5ex]
  &\betar(v)=\zeta_M(v)\circ \alphal,\quad
  \betal(v)=\zeta_M(v)\circ\alphar & &\hspace*{-3em}\text{ for all }v\in M.\label{e-betaalpha}
\end{align}

We also need to consider~$M$ as a right module over its endomorphism ring $S:=\End(\lR{M})$, see Remark~\ref{R-EndM}.
Using the isomorphism $\Theta$ from Proposition~\ref{P-RCongEndRhat} one straightforwardly
verifies that the map
\begin{equation}\label{e-tau}
   \tau: R\longrightarrow S,\quad r\longmapsto \betal^{-1}\circ\Theta(r)\circ\betal
\end{equation}
is a ring isomorphism.
Let us consider the endomorphism $\tau(r)=\betal^{-1}\circ\Theta(r)\circ\betal\in S$ for a given $r\in R$.
The definition of the maps involved yields that for $v\in M$ the image $\tilde{v}:=\tau(r)(v)$ is (uniquely defined) such that
$\chi(\sbt\tilde{v})=\chi(r\sbt v)$ as characters in~$\wR$.
In addition to the isomorphism~$\tau$ we also have the natural ring homomorphism
\begin{equation}\label{e-sigma}
  \sigma: R\longrightarrow S,\quad r\longmapsto
  \left\{\begin{array}{rcc} \sigma(r):\lR{M}&\longrightarrow&\lR{M}\\  v\ \; &\longmapsto& vr\end{array}\right.
\end{equation}
Using $\rR{M}\cong\rR{\wR}$ and Proposition~\ref{P-rwR} we conclude that~$\sigma$ is injective and thus bijective because $|R|=|S|$
thanks to Proposition~\ref{P-RCongEndRhat}.
In other words, every left $R$-linear map on~$M$ is given by right multiplication by some~$r\in R$.
As a consequence, the endomorphism ring~$S=\End(\lR{M})$ may be written in the two ways
\begin{equation}\label{e-S2}
           S=\{\tau(r)\mid r\in R\}=\{\sigma(r)\mid r\in R\}.
\end{equation}
We will need these two descriptions of~$S$ when proving the double annihilator properties later in this section.

It is interesting to note the relation between the isomorphisms~$\tau$ and~$\sigma$.
Since~$\chi$ is a left and right generating character of~$M$ there exists for every $r\in R$ a unique element $\tilde{r}\in R$
such that $r\chi=\chi\tilde{r}$.
This gives rise to the map $g:R\longrightarrow R,\ r\longmapsto g(r)$, where $g(r)$ is such that $r\chi=\chi g(r)$.
One easily verifies that~$g$ is a ring automorphism and
$\sigma=\tau\circ g$.

Symmetrically to the above, every right $R$-module~$M$ is an $(S',R)$-bimodule where $S'=\End(\rR{M})$, and
Proposition~\ref{P-RCongEndRhat} yields that for a Frobenius bimodule~$\biMod{R}{M}{\!R}$ we have
\[
           \End(\lR{M}) \cong R \cong \End(\rR{M})\text{ as rings.}
\]
In particular,~$\biMod{R}{M}{\!R}$ is balanced in the sense of \cite[Sec.~1]{GNW04}.

As a consequence of~\eqref{e-S2}, any subset $K\subseteq M$ satisfies
\begin{equation}\label{e-RightSubM}
              K\text{ is a submodule of }M_R\Longleftrightarrow K\text{ is a submodule of }M_S.
\end{equation}

Now we are ready to formulate and prove the double annihilator properties of Frobenius bimodules.
For any bimodule $\biMod{R}{M}{\!S}$ and subsets~$K\subseteq M,\,I\subseteq R,\,J\subseteq S$ we define the following annihilators
\begin{equation}\label{e-Perps}
\left.\begin{split}
  &\lp{K}:=\{r\in R\mid rK=0\}, \  &\ \rp{K}&:=\{s\in S\mid Ks=0\},\\
  &\ \rp{I}:=\{v\in M\mid Iv=0\},\ &\ \lp{J}&:=\{v\in M\mid vJ=0\}.
\end{split}\qquad\qquad
\right\}
\end{equation}
Clearly, $\lp{K},\,\lp{J}$ are submodules of~$\lR{}R$ and~$\lR{M}$, respectively,
while $\rp{K},\,\rp{I}$ are submodules of $S_{S}$ and $M_S$, respectively.

\begin{prop}\label{P-DoubleAnn}
Let $\biMod{R}{M}{\!R}$ be a Frobenius bimodule and $S:=\End(\lR{M})$.
Then for the bimodule $\biMod{R}{M}{\!S}$ we have the double annihilator properties
\[
\begin{array}{cclccl}
  (\lp{K})^\perp\!&\!=\!&\!K\text{ for all }K\subseteq M_S, \ &\mbox{}^\perp(\rp{K})\!&\!=\!&\!K \text{ for all }K\subseteq\lR{M}\\[.8ex]
  \mbox{}^\perp(\rp{I})\!&\!=\!&\!I \text{ for all }I\subseteq\lR{R}, \ &(\lp{J})^\perp\!&\!=\!&\!J\text{ for all }J\subseteq S_S.
\end{array}
\]
\end{prop}
This proposition along with \cite[Thm.~2.1(e)(ii), Prop.~2.17]{GNW04} shows that our definition of Frobenius bimodule is equivalent to the one given
in~\cite[Sec.~2]{GNW04}.

\begin{proof}
Let us first consider first the case where $M=\wR$.
By Proposition~\ref{P-RCongEndRhat} we can naturally identify~$R$ and~$S$ and thus $\biMod{R}{M}{\!S}=\biMod{R}{M}{\!R}$.
One easily verifies, see also~\cite[p.~409]{Hon01}, that
\begin{align}
  \rp{K}&=K^\circ \text{ for all }K\subseteq\lR{\wR},
   \qquad\lp{K}=K^\circ \text{ for all }K\subseteq\rR{\wR}, \label{e-Kdual}\\
  \rp{I}&=I^\circ \text{ for all }I\subseteq\lR{R},
   \qquad\quad\lp{J}=J^\circ \text{ for all }J\subseteq\rR{R},\label{e-IJdual}
\end{align}
where $\mbox{\;}^\circ$ refers to the dual groups as in Remark~\ref{R-CharGroup}(7).
Now the double annihilator properties follow from Remark~\ref{R-CharGroup}(7)(iii).

Let now~$\biMod{R}{M}{\!R}$ be an arbitrary Frobenius bimodule with the left and right isomorphisms~$\betal$ and~$\betar$
as in~\eqref{e-beta}.
We also have the ring isomorphism~$\tau$ as in~\eqref{e-tau}.
Thus, $S=\{\tau(r)\mid r\in R\}$.
We will only prove the very first identity as this is the only one that will be needed later in the paper.
The others are proven in a similar manner.
\\
Let $K\subseteq M_S$. Set $L:=\betal(K)\subseteq\wR$.
We want to show that $L\subseteq \rR{\wR}$.
This is not a priori clear as~$\betal$ is not right linear in general.
Consider for any $r\in R$ the endomorphism $\tau(r)=\betal^{-1}\circ\Theta(r)\circ\betal$.
Then for $v\in K$ we have $\tau(r)(v)=\betal^{-1}(\betal(v)r))$, thus
$\betal(v)r=\betal\big(\tau(r)(v)\big)$.
The right $S$-module structure of~$K$ implies that with $v\in K$ we also have $v\cdot\!\tau(r)=\tau(r)(v)\in K$ and therefore
$\betal\big(\tau(r)(v)\big)\in L$.
Hence $\betal(v)r\in L$ for any $v\in K$ and $r\in R$, and this establishes that $L$ is a submodule of $\rR{\wR}$.
For the rest of the proof we use the notation $I^{\perp_M}$ and $I^{\perp_{\wR}}$ in order to distinguish between the
annihilators of~an ideal $I\subseteq\lR{R}$ in~$M$ and in~$\wR$.
Then left linearity and injectivity of~$\betal$ implies $\betal(I^{\perp_M})=I^{\perp_{\wR}}$  for any $I\subseteq \lR{R}$.
By the same properties we have $\lp{K}=\lp{L}$.
Now the special case $M=\wR$ gives us $(\lp{L})^{\perp_{\wR}}=L$ and thus we compute
\[
  \big|(\lp{K})^{\perp_M}\big|=\big|\betal\big((\lp{K})^{\perp_M}\big)\big|
  =\big|\betal\big((\lp{L})^{\perp_M}\big)\big|
  =\big|(\lp{L})^{\perp_{\wR}}\big|=|L|=|K|.
\]
Together with the obvious containment $K\subseteq(\lp{K})^{\perp}$ we obtain the desired identity.
\end{proof}

The equivalence in~\eqref{e-RightSubM} along with the fact that for $K\subseteq M$ the double annihilator $(\lp{K})^\perp$ does not involve
any right module structure of~$M$ yields the following consequence.
\begin{cor}\label{C-KMR}
For any Frobenius bimodule $\biMod{R}{M}{\!R}$ and any $K\subseteq\rR{M}$ we have $(\lp{K})^\perp=K$.
\end{cor}

\section{Dualizations of Partitions of Finite Frobenius Bimodules}\label{S-Partitions}
In this section we derive a character-theoretic dualization of partitions of Frobenius bimodules.
This generalizes the results in~\cite{GL15Fourier}  where partitions of~$R^n$ with~$R$ a Frobenius ring were considered.

Let us first set up some general notation.
Let $\cP = (P_k)_{k=1}^K$ be a partition of a finite set~$X$, i.e., $X$ is the disjoint union of the subsets $P_1,\ldots,P_K$.
The sets $P_k$ are called $\emph{blocks}$, and $|\cP| = K$ is the number of blocks of $\cP$ (assuming all blocks are nonempty).
A partition~$\cP$ is called \emph{finer} than the partition~$\cQ$ if every block of $\cP$ is contained in some block of $\cQ$.
In this case we write $\cP\leq\cQ$, and it follows that $|\cQ| \leq |\cP|$.
Two partitions~$\cP,\,\cQ$ of~$X$ are called \emph{identical} if $\cP \leq \cQ$ and
$\cQ\leq\cP$.
We will denote by $\sim_{\cP}$ the equivalence relation induced by the partition $\cP$.

We now turn to modules~$V^n$ where $\biMod{R}{V}{\!R}$ is a bimodule over the finite ring~$R$.
Then clearly $V^n:=\{(v_1,\ldots,v_n)\mid v_i\in V\}$ is an $(R,R)$-bimodule in the natural way.
We will consider $v\in V^n$ as a row vector and otherwise write $v\T$.
From Remark~\ref{R-CharGroup}(2) we obtain an $(R,R)$-bilinear isomorphism
\begin{equation}\label{e-Vhat}
    \widehat{V}^{\,n}\cong \widehat{V^n}\ \text{ via }\ (\chi_1,\ldots,\chi_n)(v_1,\ldots,v_n):=\prod_{i=1}^n\chi_i(v_i),
\end{equation}
and we will simply identify these modules. In particular, $\wR^{\,n}=\widehat{R^n}$.
For matrices $A\in R^{n\times m}$ and $B\in R^{m\times n}$, where $m\in\N$, and for $v\in V^n$ we have
$vA,\,(Bv\T)\T\in V^m$ in the obvious way.
This gives rise to the following group actions. Its orbits will play a crucial role later on.

\begin{defi}\label{D-UAction}
Let~$\cU$ be a subgroup of $\GL_n(R)$.
Then~$\cU$ induces a right and left group action on~$V^n$ via
\[
  V^n\times\cU\longrightarrow V^n,\
  (v, U) \longmapsto vU\ \text{ and }\
  \cU\times V^n\longrightarrow V^n,\
  (U, v) \longmapsto (Uv\T)\T.
\]
Denote by $\cP_{V^n,\,\cU}$ and $\cP_{V^n,\,\cU^\top}$ the respective orbit partitions on~$V^n$.
\end{defi}

From now on, let $R$ be a finite ring and $\biMod{R}{M}{\!R}$ be a finite Frobenius bimodule with generating character~$\chi$.
Note that $|M^n|=|\wM^n|=|R^n|=|\wR^n|$.
Let $\ideal{\,\sbt,\,\sbt}$ denote both the standard dot products
$M^n\times R^n \longrightarrow M,\ \ideal{v,r}:=v r\T=\sum_{i=1}^n v_i r_i$ and
$R^n\times M^n \longrightarrow M,\ \ideal{r,v}:=rv\T=\sum_{i=1}^n r_i v_i$.
Thanks to Proposition~\ref{P-rwR} these bilinear forms are non-degenerate.
One easily verifies that the maps $\betal,\,\betar,\,\alphal,\,\alphar$ from~\eqref{e-beta}, \eqref{e-alpha}  generalize to the isomorphisms
\begin{align}
  &\beta_l : \lR{(M^n)} \longrightarrow \lR{(\wR^n)},\ v \longmapsto \chi(\ideal{\, \sbt, v})\ \text{ and }&
   \beta_r : \rR{(M^n)} \longrightarrow \rR{(\wR^n)},\ v \longmapsto \chi(\ideal{v,\, \sbt}) \label{e-betaleftright}\\[.5ex]
   &\alpha_l : \lR{(R^n)} \longrightarrow \lR{(\wM^n)},\ r \longmapsto \chi(\ideal{\,\sbt, r}) \text{ and }&
   \alpha_r : \rR{(R^n)}  \longrightarrow \rR{(\wM^n)},\ r \longmapsto \chi(\ideal{r,\, \sbt}) \label{e-alphaleftright}
\end{align}
Note that
\[
  \alphal(r)(v)=\chi(\ideal{v,r})=\chi({\TS\sum_{i=1}^n v_i r_i})=\prod_{i=1}^n (r_i\chi)(v_i)
\]
and thus with the identification~\eqref{e-Vhat} we obtain the isomorphism
\begin{equation}\label{e-alphal1}
  \alpha_l:\lR{(R^n)}\longrightarrow \lR{(\wM^n)},\ r\longmapsto (r_1\chi,\ldots,r_n\chi)
\end{equation}
and similarly
\begin{equation}\label{e-alphar1}
  \alpha_r:\lR{(R^n)}\longrightarrow \lR{(\wM^n)},\ r\longmapsto (\chi r_1,\ldots,\chi r_n)
\end{equation}
In the same way we have $\betal(v)=(\chi(\sbt v_1),\ldots,\chi(\sbt v_n))$ for all $v\in R^n$ and similarly for~$\betar$.

The isomorphisms in~\eqref{e-betaleftright} and~\eqref{e-alphaleftright} satisfy the simple relations
\begin{equation}\label{e-alphabeta2}
   \alphal(r)(v)=\betar(v)(r)\ \text{ and }\ \alphar(r)(v)=\betal(v)(r)\ \text{ for all }r\in R^n,\,v\in M^n;
\end{equation}
see also~\eqref{e-alphabeta}.
These isomorphisms will be crucial for developing a duality theory for partitions.

\begin{defi}\label{D-DualPartGroup}
Let~$A$ be a finite abelian group and $\cP = (P_k)_{k=1}^K$ be a partition of $A$.
The \emph{dual partition} of $\cP$, denoted by~$\widehat{\cP}$, is the partition of $\wA$ defined via the equivalence relation
\[
  \psi \sim_{\widehat{\cP}} \psi' \Longleftrightarrow
         \sum_{a\in P_k} \psi (a) = \sum_{a\in P_k} \psi'(a) \ \text{ for all }\  k= 1, \ldots , K.
\]
We call~$\cP$ \emph{reflexive} if $\cP=\wwcP$ (where we identify $A$ and~$\bidual{A}$).
\end{defi}

\begin{rem}\label{R-DualPartGroup}
It is well-known (see \cite[Thm.~2.4]{GL15Fourier} and \cite[Fact~V.2]{Hon10}) that $|\cP|\leq|\wcP|$ and $\wwcP\leq\cP$.
Furthermore, $|\cP|=|\wcP|\Longleftrightarrow\cP=\wwcP$.
\end{rem}

We turn now to partitions of~$M^n$ and $R^n$.
Using the bilinear forms above we can define specific left and right dual partitions in $R^n$ and $M^n$, respectively.
\begin{defi}\label{D-DualPartMod}
For a partition $\cP = (P_k)_{k=1}^K$ of~$R^n$ the $\chi$-\emph{left dual} and $\chi$-\emph{right dual partitions} are
 the partitions of~$M^n$ defined via
\[
   v\widesim_{\wcPchil} v'
   \Longleftrightarrow
    \sum_{r\in P_k} \chi(\ideal{r,v}) = \sum_{r\in P_k} \chi(\ideal{r,v'})\text{ for all }k = 1, \ldots, K,
\]
and
\[
   v\widesim_{\wcPchir} v' \Longleftrightarrow
    \sum_{r\in P_k} \chi(\ideal{v,r}) = \sum_{r\in P_k} \chi(\ideal{v',r})\text{ for all }k = 1, \ldots, K.
\]
Similarly, for a partition $\cQ = (Q_k)_{k=1}^L$ of~$M^n$
the $\chi$-\emph{left dual} and $\chi$-\emph{right dual partitions} are the partitions of~$R^n$
defined by the equivalence relations
\[
   r\widesim_{\wcQchil} r'
   \Longleftrightarrow
    \sum_{v\in Q_k} \chi(\ideal{v, r}) = \sum_{v\in Q_k} \chi(\ideal{v, r'})\text{ for all }k = 1, \ldots, L,
\]
and
\[
   r \widesim_{\wcQchir} r' \Longleftrightarrow
    \sum_{v\in Q_k} \chi(\ideal{r,v}) = \sum_{v\in Q_k} \chi(\ideal{r',v})\text{ for all }k = 1, \ldots, L.
\]
\end{defi}
The so defined dual partitions do indeed depend on the choice of the generating character~$\chi$.
But this can easily be described. Suppose $\chi'$ is another generating character of~$M$.
Then $\chi'=u\chi=\chi\tilde{u}$ for some units $u,\,\tilde{u}\in R^*$, see Remark~\ref{R-GenCharMod}.
As a consequence, $\chi'(\ideal{r,v})=\chi(\ideal{r,vu})$ and therefore for any partition~$\cP$ of $R^n$ we have
$v\widesim_{\wcP^{[\chi',l]}} v'\Longleftrightarrow vu\widesim _{\wcP^{[\chi,l]}} v'u$ and thus
$\wcP^{[\chi',l]}u=\wcP^{[\chi,l]}$, where the latter means that each block is right multiplied by~$u$.
In the same way $\tilde{u}\wcP^{[\chi',r]}=\wcP^{[\chi,r]}$.
Analogous relations hold true for the duals of partitions of~$M^n$.

\begin{rem}\label{R-DualPart}\
\begin{alphalist}
\item For a partition~$\cP$ of $R^n$ we have $\wcPchil:=\betal^{-1}(\widehat{\cP})$  and $\wcPchir:=\betar^{-1}(\widehat{\cP})$.
\item For a partition~$\cQ$ of $M^n$ we have $\wcQchil:=\alphal^{-1}(\widehat{\cQ})$  and $\wcQchir:=\alphar^{-1}(\widehat{\cQ})$.
\end{alphalist}
As a consequence, $|\wcPchil|=|\wcPchir|=|\wcP|\geq|\cP|$ for each partition~$\cP$ of~$R^n$ or~$M^n$.
\end{rem}

The Hamming weight gives rise to a particularly nice partition.
\begin{exa}\label{E-HammPart}
On~$R^n$ and~$M^n$ consider the Hamming weight $\wtH(v_1,\ldots,v_n)=|\{i\mid v_i\neq0\}|$ (for a character-module this reads as
$\wtH(\psi_1,\ldots,\psi_n)=|\{i\mid \psi_i\neq\varepsilon\}|$).
This gives rise to the Hamming partition~$\cP_{M^n,\text{Ham}}=(P_i)_{i=0}^n$, where
$P_i=\{v\in M^n\mid \wtH(v)=i\}$.
Then $\widehat{\cP_{M^n,\text{Ham}}}^{[\chi,l]}=\widehat{\cP_{M^n,\text{Ham}}}^{[\chi,r]}=\cP_{R^n,\text{Ham}}$, the Hamming partition on~$R^n$.
This follows from the well-known group-theoretic case, see e.g.,~\cite[Ex.~2.3(c)]{GL15Fourier}, along with the fact that~$\alphal$
and~$\alphar$ are Hamming-weight-preserving isomorphisms as can be seen from~\eqref{e-alphal1} and~\eqref{e-alphar1}.
Now, $\widehat{\cP_{R^n,\text{Ham}}}^{[\chi,l]}=\widehat{\cP_{R^n,\text{Ham}}}^{[\chi,r]}=\cP_{M^n,\text{Ham}}$ follows from reflexivity
of the Hamming partition, which in turn is a trivial consequence of Remark~\ref{R-DualPartGroup}.
\end{exa}

Remark~\ref{R-DualPart} allows us to prove the following analogue of \cite[Prop.~4.4]{BGL15}.

\begin{thm}\label{T-chiBiDual}
Let $\cP$ be any partition of $R^n$ or $M^n$.
Then
\[
   \widehat{\wcPchil}^{[\chi,r]}=\wwcP=\widehat{\wcPchir}^{[\chi,l]},
\]
where~$\wwcP$ is the bidual in the sense of Definition~\ref{D-DualPartGroup}.
As a consequence,
\[
     \cP\text{ is reflexive}\Longleftrightarrow\cP= \widehat{\wcPchil}^{[\chi,r]}\Longleftrightarrow\cP= \widehat{\wcPchir}^{[\chi,l]}.
\]
\end{thm}
\begin{proof}
Let $\cP$ be a partition of $M^n$. Set $\cQ=(Q_k)_{k=1}^N=\wcPchil$ and $\cR=\wcQchir=\betar^{-1}(\wcQ)$.
Let $v,\,v'\in M^n$. With~\eqref{e-alphabeta2}  we compute
\begin{align*}
     v\widesim_{\cR}v' &\Longleftrightarrow \betar(v)\sim_{\wcQ}\betar(v')\\
    &\Longleftrightarrow \sum_{r\in Q_k}\betar(v)(r)=\sum_{r\in Q_k}\betar(v')(r)\ \text{ for all }k=1,\ldots,N\\
    &\Longleftrightarrow \sum_{r\in Q_k}\alphal(r)(v)=\sum_{r\in Q_k}\alphal(r)(v')\ \text{ for all }k=1,\ldots,N\\
    &\Longleftrightarrow \sum_{\Psi\in\alpha_{\rm l}(Q_k)}\!\!\Psi(v)= \sum_{\Psi\in\alpha_{\rm l}(Q_k)}\!\!\Psi(v')\ \text{ for all }k=1,\ldots,N\\
    &\Longleftrightarrow v\widesim_{\widehat{\alpha_{\rm l}(\cQ)}} v' \\
    &\Longleftrightarrow v\widesim_{\mbox{$\scriptsize\widehat{\phantom{\Big|}\hspace*{.5em}}\hspace*{-1em}\wcP$}}\, v'.
\end{align*}
This establishes $\widehat{\wcPchil}^{[\chi,r]}=\wwcP$. The other one as well as those for partitions of~$R^n$ are shown in the same way.
The rest follows.
\end{proof}

For the orbit partitions of group actions (see Definition~\ref{D-UAction}) the above leads to the following relations.

\begin{lem}\label{L-UAction}
Let $\cU$ be a subgroup of $\GL_n(R)$. Then
\begin{alphalist}
\item $\cP_{R^n,\,\cU}\leq\widehat{\cP_{M^n,\,\cU^{\sf T}}}^{[\chi,r]}$ and $\cP_{R^n,\,\cU^{\sf T}}\leq\widehat{\cP_{M^n,\,\cU}}^{[\chi,l]}$.
\item $\cP_{M^n,\,\cU}\leq\widehat{\cP_{R^n,\,\cU^{\sf T}}}^{[\chi,r]}$ and $\cP_{M^n,\,\cU^{\sf T}}\leq\widehat{\cP_{R^n,\,\cU}}^{[\chi,l]}$.
\end{alphalist}
\end{lem}

\begin{proof}
For brevity set $\cP:=\cP_{R^n,\,\cU},\,\cQ:=\cP_{M^n,\, \cU^{\sf T}}$.
Let $r,\,r'\in R^n$ such that $r\widesim_{\cP} r'$, thus $r'=rU$ for some $U\in\cU$.
Then for any $v\in M^n$ we have $\ideal{r',v}=\ideal{rU,v}=\ideal{r,(Uv\T)\T}$.
Let $Q$ be any block of $\cQ$. Then the closedness of~$Q$ under the left action of~$\cU$ yields
\[
   \sum_{v\in Q}\chi(\ideal{r',v})=\sum_{v\in Q}\chi(\ideal{r,(Uv\T)\T})=\sum_{v\in Q}\chi(\ideal{r,v}).
\]
This shows $r\widesim_{\widehat{\cQ}^{[\chi,r]}}r'$, as desired. The other relations are shown in the same way.
\end{proof}

Now we obtain
\begin{thm}\label{T-UPart}
Let~$\cU$ be a subgroup of $\GL_n(R)$. Then
\[
   \cP_{R^n,\,U}=\widehat{\cP_{M^n,\,\cU^{\sf T}}}^{[\chi,r]},\
   \cP_{R^n,\,\cU^{\sf T}}=\widehat{\cP_{M^n,\,\cU}}^{[\chi,l]},\
   \cP_{M^n,\,\cU}=\widehat{\cP_{R^n,\,\cU^{\sf T}}}^{[\chi,r]},\
   \cP_{M^n,\,\cU^{\sf T}}=\widehat{\cP_{R^n,\,\cU}}^{[\chi,l]}.
\]
As a consequence, all these partitions are reflexive.
Moreover, the right (resp.\ left) group action of~$\cU$ on~$R^n$ and
the left (resp.\ right) action of~$\cU$ on~$M^n$ lead to the same number of orbits.
\end{thm}
\begin{proof}
Combining Remark~\ref{R-DualPart} and Lemma~\ref{L-UAction} we obtain
\[
  |\cP_{R^n,\,U}|\geq|\widehat{\cP_{M^n,\,\cU^{\sf T}}}^{[\chi,r]}|\geq |\cP_{M^n,\,\cU^{\sf T}}|\geq|\cP_{R^n,\,U}|.
\]
Thus we have equality everywhere, and again with Lemma~\ref{L-UAction} we arrive at the first identity.
The others are shown in the same way.
\end{proof}

\begin{exa}\label{E-NonFrob}
Let $R=\F_2[x,y]/(x^2,y^2,xy)$, which is a commutative non-Frobenius ring, and let $M$ be the Frobenius bimodule $M=\wR$.
Furthermore, let
\[
  \cU=\bigg\{\begin{pmatrix}1&r\\0&u\end{pmatrix} \bigg|\, r\in R,\,u\in R^*\bigg\}.
\]
Then one can show that both, $\cP_{R^2,\,\cU}$ and $\cP_{M^2,\,\cU^{\sf T}}$, consist of~$17$ orbits whereas
$\cP_{R^2,\,\cU^{\sf T}}$ and $\cP_{M^2,\,\cU}$ consist of~$20$ orbits.
This shows that for non-Frobenius rings the right and left action of~$\cU$ do not lead in general to the same
number of orbits (see also \cite[Sec.~4]{BGL15}).
Using reflexivity of all partitions involved, the just mentioned cardinalities are compliant with the previous theorem.
\end{exa}

We make use of the last result in the following technical lemma which will be crucial in the next section.

\begin{lem}\label{L-AiInner}
Let $\cC \subseteq \lR{M^n}$ and let $\cU$ be a subgroup of $\GL_n(R)$.
Assume $f:\cC\longrightarrow M^n$ is a left linear map such that for all $x\in\cC$ there exists a matrix $U_x\in \cU$ such that $f(x) = xU_x$.
Then, for all $r\in R^n$ there exists a matrix $A_r \in \cU$ such that $\ideal{ f(x),\, r} = xA_r r^{\sf T}$ for all $x\in\cC$.
\end{lem}

\begin{proof}
Let $P$ be a block of $\cP_{R^n,\,\cU^{\sf T}}=\widehat{\cP_{M^n,\,\cU}}^{[\chi,l]}$.
Then for all $x \in \cC$
\[
    \sum_{r\in P} \chi(\ideal{ f(x),\, r}) = \sum_{r\in P} \chi(\ideal{xU_x,\, r})
    = \sum_{r\in P} \chi(\ideal{ x, (U_x r^{\sf T})^{\sf T}}) = \sum_{r\in P} \chi(\ideal{ x,\,r}),
\]
where the last step follows from the invariance of~$P$ under the left action of~$\cU$.
As a consequence,
\[
     \sum_{r\in P} \chi(\ideal{ f(\sbt), \,r}) = \sum_{r\in P} \chi(\ideal{\sbt,\, r})
\]
and  each side of the identity is a sum of  elements in the character group $\widehat{\cC}$.
Fix $r\in R^n$ and assume that $r$ is contained in the block $P$ of $\cP_{R^n,\,\cU^{\sf T}}$.
Remark~\ref{R-CharGroup}(6) implies that the character $\chi(\ideal{f(\sbt),\, r })$ must appear on the right hand side of the above identity.
In other words, there exists $r'\in P$ i.e., $r' = (A_rr^{\sf T})^{\sf T}$ for some $A_r \in \cU$, such that $\chi(\ideal{ f(\sbt),\, r}) = \chi(\ideal{\sbt,\, r'})$.
With the aid of Theorem~\ref{T-FrobBimod} we conclude that $\ideal{ f(\sbt),\, r}= \ideal{\sbt,\, r'}$ as maps in $\text{Hom}(\lR{\cC},\lR{M})$.
This implies $\ideal{f(x),\, r} = \ideal{ x,\, (A_rr^{\sf T})^{\sf T}} = xA_rr^{\sf T}$, as desired.
\end{proof}

\section{MacWilliams Extension Theorem for Frobenius Bimodules}\label{S-Extensions}

In this section we establish MacWilliams Extension Theorem for codes over Frobenius bimodules with respect to various weight functions.

Throughout, let~$R$ be any finite ring and $\lR{M}$ be a finite module.
A map $\wt:M\longrightarrow\C$ satisfying $\wt(0)=0$ is called a \emph{weight function} on~$M$.
Note that a weight function on~$M$ can naturally be extended to a weight function on~$M^n$ via
\begin{equation}\label{e-AddWeights}
  \wt(v_1,\ldots,v_n):=\sum_{i=1}^n \wt(v_i)\ \text{ for all }\ (v_1,\ldots,v_n)\in M^n.
\end{equation}

\begin{exa}\label{E-HamWt}
The \emph{Hamming weight with respect to the alphabet~$M$} is defined as the weight function on~$M^n$, where $n\in\N$, given by
$\wtH(v_1,\ldots,v_n):=|\{i\mid v_i\neq0\}|$. It may be regarded as the extension of the
trivial weight $\wt(v)=\delta_{v,0}$ for $v\in M$, where $\delta$ denotes the Kronecker delta function.
The latter itself is the Hamming weight on~$M$ with respect to the alphabet~$M$.
\end{exa}

There are, of course, also weight functions on~$M^n$ that do not arise as such an extension.
One such weight is the Rosenbloom-Tsfasman weight, which is a special case of general poset weights; see also Remark~\ref{R-OtherWeights}(d).
This weight has been introduced by Rosenbloom and Tsfasman in~\cite{RoTs97} and plays a specific role for matrix
codes; see for instance~\cite{Skr07} for the relevance of the Rosenbloom-Tsfasman weight for detecting matrix codes
with large Hamming distance.

\begin{defi}\label{D-RTWeight}
The Rosenbloom-Tsfasman weight (RT-weight) of a vector $v = (v_1, \ldots, v_n) \in M^n$ is defined as
\[
  \wtRT(v) = \left\{\begin{array}{cl}
      0, & v = 0, \\[.5ex]
      \text{max} \{i\mid v_i \neq 0\}, & \text{otherwise.}
   \end{array}\right.
\]
\end{defi}

An important role in ring-linear coding is played by the homogeneous weight.
\begin{defi}\label{D-homogWt}
A weight function $\omega:M\longrightarrow\C$ on a finite module~$\lR{M}$ is called \emph{(normalized left) homogeneous} if
\begin{romanlist}
\item $\omega(v)=\omega(w)$ for all $v,\,w\in M$ such that $Rv=Rw$.
\item $\sum_{w\in Rv}\omega(w)=|Rv|$ for all $v\neq 0$.
\end{romanlist}
\end{defi}
In~\cite[Thm.~4.4]{GNW04} Greferath et al.\ establish the existence and uniqueness of the homogeneous weight on arbitrary finite modules.
For finite Frobenius bimodules a very useful formula for the homogeneous weight has been established by Wood~\cite[Prop.~9]{Wo14}.
It is a straightforward generalization of \cite[p.~412]{Hon01} by Honold, where the same result was derived for finite Frobenius rings.
The proof follows from observing that $\sum_{\alpha\in R^*}\alpha\chi=\sum_{\alpha\in R^*}\chi\alpha$ is the sum of all generating
characters of~$M$ and by verifying~(i) and~(ii) of Definition~\ref{D-homogWt}.

\begin{thm}\label{T-HomogChar}
Let $\biMod{R}{M}{\!R}$ be a finite Frobenius bimodule with generating character~$\chi$.
Then the homogeneous weight on~$M$ is given by
\[
    \omega(v)=1-\frac{1}{|R^*|}\sum_{\alpha\in R^*}\chi(v\alpha)=1-\frac{1}{|R^*|}\sum_{\alpha\in R^*}\chi(\alpha v)\ \text{ for all }v\in M.
\]
\end{thm}

Of course, the homogeneous weight depends highly on the module structure. This is shown in the first example.
\begin{exa}\label{E-HomogRingMod}
\begin{alphalist}
\item \cite[Ex.~3.4(b)]{GL14homog} Consider the Frobenius ring $R=\Z_2\times\Z_2$.
      Then the homogeneous weight on~$R$ is given by $\omega(0,0)=0=\omega(1,1)$ and $\omega(0,1)=\omega(1,0)=2$.
      In particular, there are nonzero elements with zero weight.
      Since $\{0\}$ is a block of the dual of any partition \cite[Rem.~2.2(a)]{GL15Fourier} this implies
      that the induced partition $\cP_{\text{hom}}=(P_i)_{i=0,2}$, where $P_i=\{r\in R\mid \omega(r)=i\}$ is not reflexive.
      On the other hand, the homogeneous weight on the $\Z_2$-module $M=\Z_2^2$ (which is of course not a Frobenius bimodule) is easily
      seen to be given by $\hat{\omega}(v)=2$ for all $v\in M\backslash\{(0,0)\}$.
      Finally, if we take the homogeneous weight on the ring~$\Z_2$, which is simply the Hamming weight, and extend it to~$M$ as
      in~\eqref{e-AddWeights}, we obtain $\tilde{\omega}(1,0)=\tilde{\omega}(0,1)=1$ and $\tilde{\omega}(1,1)=2$.
\item Let $R$ be the non-Frobenius ring $R=\F_2[x,y]/\ideal{x^2,xy,y^2}=\{a+bx+cy\mid a,b,c\in\F_2\}$ computing modulo $\ideal{x^2,xy,y^2}$.
       Its group of units is $R^*=\{1,1+x,1+y,1+x+y\}$.
      Consider the Frobenius bimodule $M:=\wR$.
      We compute $M=\spann_{\F_2}\{\chi_1,\chi_2,\chi_3\}$, where
      \begin{align*}
         &\chi_1(x)=\chi_1(y)=1,\,\chi_1(1)=-1,\\
         &\chi_2(1)=\chi_2(y)=1,\,\chi_2(x)=-1,\\
         &\chi_3(1)=\chi_3(x)=1,\,\chi_3(y)=-1.
      \end{align*}
      This determines the remaining values of these and all other characters in~$M$.
      Using Remark~\ref{R-GenCharMod} we obtain the generating character $\psi$ of~$M$ defined as $\psi(\chi)=\chi(1)$.
      Now Theorem~\ref{T-HomogChar} yields the homogeneous weight $\omega(\chi)=1-\frac{1}{4}(\chi(1)+\chi(1+x)+\chi(1+y)+\chi(1+x+y))$, thus
      \[
        \omega(\varepsilon)=0,\ \omega(\chi_1)=2, \ \omega(\chi)=1\text{ for all }\chi\in M\,\backslash\,\{\varepsilon,\chi_1\}.
      \]
      This is also obtained by observing that
      \[
          R\chi_1=\{\varepsilon,\,\chi_1\}\ \text{ and }\ R\chi=\{\varepsilon,\,\chi,\,\chi_1,\,\chi+\chi_1\} \text{ for all }
          \chi\in M\,\backslash\,\{\varepsilon,\chi_1\}.
      \]
\end{alphalist}
\end{exa}

We now turn to weight-preserving maps.

\begin{defi}\label{D-code}
Let~$\lR{M}$ be any finite module.
A submodule of $\lR{(M^n)}$ is called a \emph{code over the alphabet $M$ of length $n$}.
\end{defi}

\begin{defi}\label{D-WeightFct}
Let $\wt$ be a weight function on the finite module~$\lR{M}$ and let $\cC\subseteq\lR{M^n}$ be a code over~$M$ of length~$n$.
A linear map $f:\cC\longrightarrow M^n$ is called \emph{$\wt$-preserving} if $\wt(v)=\wt(f(v))$ for all $v\in\cC$.
If~$f$ is also injective it is called an \emph{$\wt$-isometry}.
\end{defi}
Note that if the weight function satisfies $\wt(x)=0\Longleftrightarrow x=0$, then any $R$-linear $\wt$-preserving map from~$\cC$ to~$M^n$
is a $\wt$-isometry.
In particular, any $R$-linear Hamming-weight preserving map or RT-weight preserving map is an isometry for the respective weight.

From now on let~$M$ be finite Frobenius bimodule with generating character~$\chi$ over the finite ring~$R$.
Then the above isometries on~$\lR{M^n}$ can easily be described.
Recall from~\eqref{e-sigma} that $R\cong S:=\End(\lR{M})$ naturally.
As a consequence, we have an induced group isomorphism $R^*\cong\Aut(\lR{M})$.
Furthermore, one obtains straightforwardly the ring isomorphism
\begin{equation}\label{e-EndMn}
   \cM_n(R)\cong \End(\lR{M^n}), \quad A\longmapsto
   \left\{\begin{array}{ccc}\lR{M^n}&\longrightarrow&\lR{M^n}\\ v&\longmapsto& vA\end{array}\right.
\end{equation}
where $\cM_n(R)$ denotes the ring of $n\times n$-matrices with entries in~$R$.
As a consequence we obtain $\GL_n(R)\cong\Aut(\lR{M^n})$.

Denote the group of invertible lower triangular matrices as
\begin{equation}\label{e-LTR}
   \LT_n(R):=\{A\in\GL_n(R)\mid A\text{ is lower triangular}\}.
\end{equation}
Furthermore, we define the group of monomial matrices over~$R$ as
\[
   \Mon_n(R):=\{A\in\GL_n(R)\mid A \text{ has exactly one nonzero entry in each row and column}\}.
\]
Note that by invertibility the nonzero elements in~$A$ are units.
We also define, for a subgroup $\cU\leq\GL_n(R)$, the group of $\cU$-monomial matrices as
\[
  \Mon_{\,\cU,n}(R):=\{A\in\Mon_n(R)\mid \text{the nonzero entries of~$A$ are in~$\cU$}\}.
\]

\begin{prop}\label{P-Isom}
Let $f\in\End(\lR{M^n})$ and $A\in\cM_n(R)$ be such that $f(v)=vA$ for all $v\in M^n$.
\begin{alphalist}
\item $f$ is a $\wtH$-isometry iff $A\in\Mon_n(R)$.
\item $f$ is a $\wtRT$-isometry iff $A\in\LT_n(R)$.
\end{alphalist}
\end{prop}
\begin{proof}
In both cases the `if-part' is clear.
For the `only-if part' note first that in both cases~$f$ is an isometry because both weights are
zero only for the zero vector.
Thus $A\in\GL_n(R)$.
\\
(a)  Consider the restrictions of~$f$ to the submodules
$\pi_i(M^n)=\{(0,\ldots,0,v,0,\ldots,0)\mid v\in M\}$, where the nonzero element is at the $i$th position.
These restrictions induce isomorphisms from $\pi_i(M^n)$ to some $\pi_j(M^n)$, which in turn implies that
$A_{ij}\in R^*$ whereas all other entries in the $i$th row of~$A$ are zero.
Using once more the bijectivity of~$f$ we conclude that also each column of~$A$ has only one nonzero entry.
\\
(b) In this case the result follows by considering consecutively the restrictions
to the submodules $\pi_i(M^n)$ for $i=1,\ldots,n$.
\end{proof}

We now turn to the MacWilliams extension property.
The following definition deviates from other uses in the literature in the sense that the length of the codes
is fixed through the module~$V$.
This is necessary so that we can also deal with weights that do not arise as the
extension of a weight on the alphabet as in~\eqref{e-AddWeights}, such as the RT-weight.

\begin{defi}\label{D-MacWExt}
Let $\wt$ be a weight function on a module $\lR{V}$. Then $\wt$ \emph{satisfies the MacWilliams extension
property} if for any code $\cC\subseteq\lR{V}$ each $\wt$-preserving linear map $f:\cC\longrightarrow V$
can be extended to a $\wt$-preserving linear map on~$V$.
\end{defi}

The above definition covers various cases just for the Hamming weight alone.
For example, if $V=M^n$ is endowed with the Hamming weight with respect to the alphabet~$M$,
the above is the classical MacWilliams extension property (for module alphabets).
If, however, we consider the Hamming weight on~$V$ with respect to the alphabet~$V$, then the weight-preserving maps
are exactly the injective maps and the extension property asks for extending injective maps on submodules
to injective maps on the entire module.
This last case can easily be dealt with.

\begin{prop}\label{P-HammMn}
Let $\cC$ be a submodule of~$\lR{M^n}$.
Then any injective map $f:\cC\longrightarrow M^n$ extends to an automorphism on~$M^n$.
\end{prop}
Using the fact that any Frobenius bimodule~$M$ is injective~\cite[Thm.~2.1]{GNW04}, and thus so is~$M^n$, the
result follows easily from \cite[Prop.~5.1]{Wo09} by Wood which in turn is based on the case for codes over finite rings in
\cite[Prop.~3.2]{DiLP04} by Dinh/L{\'o}pez-Permouth and uses the properties of the socle of a pseudo-injective module.
Since our exposition is fully character-theoretic we wish to add an alternative proof.
It is based on the double-annihilator property that we derived in Corollary~\ref{C-KMR}.

Before doing so, we present the following simple lemma.
It will be useful for several situations later on.
Recall that for a ring~$S$ we denote its group of units by~$S^*$.
\begin{lem}\label{L-Local}
Let $\rR{M}$ be any finite module.
Let $\cC,\,\cC'\subseteq M^n$ and $f:\cC\longrightarrow \cC'$ a bijective map.
Suppose there exists a subring $S$ of the matrix ring $\cM_n(R)$ with the property
\begin{romanlist}
\item for all $v\in\cC$ there exists a matrix $A_v\in S$ such that $f(v)=vA_v$,
\item for all $w\in\cC'$ there exists a matrix $B_w\in S$ such that $f^{-1}(w)=wB_w$.
\end{romanlist}
Then there exists for every $v\in\cC$ a matrix $C_v\in S^*$ such that $f(v)=vC_v$.
\end{lem}
\begin{proof}
Let $f(v)=w$. Then by assumption $w=vA_v$ and $v=wB_w$, and therefore the right $S$-modules $wS$ and $vS$ coincide.
Now the statement follows from Theorem~\ref{T-LeftSim}.
\end{proof}

Now we provide an alternative proof of Proposition~\ref{P-HammMn}.

\noindent {\it Proof of Proposition~\ref{P-HammMn} (following the local-global principle of~\cite{BGL15}).}
Let $f:\cC\longrightarrow M^n$ be an injective linear map.
\\
1) Let $f(x)=y$ for $x=(x_1,\ldots,x_n),\,y=(y_1,\ldots,y_n)$.
Set  $I:=x_1R+\ldots+x_n R,\,J:=y_1R+\ldots+y_nR\subseteq \rR{M}$.
Then injectivity of~$f$ implies $\lp{I}=\lp{J}$ and thus $I=J$ thanks to Corollary~\ref{C-KMR}.
\\
2) Let again $f(x)=y$. Then by~1) $y_j\in\sum_{i=1}^n x_i R$ and $x_j\in\sum_{i=1}^n y_i R$  for all~$j$.
This leads to matrices $A_x,\,B_y\in S:=\cM_n(R)$ such that $y=xA_x$ and
$f^{-1}(y)=yB_y$ for all $x\in\cC$ and $y\in f(\cC)$ .
Now Lemma~\ref{L-Local} yields matrices~$A_x\in S^*=\GL_n(R)$ such that $f(x)=xA_x$ for all $x\in\cC$.
\\
3) Lemma~\ref{L-AiInner} implies the existence of matrices $A_i\in\GL_n(R)$ such that $\ideal{f(x),e_i}=xA_ie_i\T$ for all $i=1,\ldots,n$.
Set $A:=(A_1e_1\T,\ldots,A_n e_n\T)\in S$. Then by construction $f(x)=xA$ for all $x\in\cC$.
In the same way we obtain a matrix $B\in S$ such that $f^{-1}(y)=yB$ for all $y\in f(\cC)$.
Let now $v_1,\ldots,v_k\in M^n$ be generators of~$\lR{\cC}$.
Consider $M^{k\times n}$, the set of all matrices with entries in~$M$.
It is a right $S$-module in the obvious way.
Let $G\in M^{k\times n}$ be the matrix whose rows are $v_1,\ldots,v_k$ and set $GA=G'$.
Then by the above $G'B=G$.
Thus the right $S$-modules $GS$ and $G'S$ coincide, and thanks to Theorem~\ref{T-LeftSim} there exists a matrix
$U\in S^*=\GL_n(R)$ such that $GU=G'$.
This finally shows that $f(x)=xU$ for all $x\in\cC$.
Thus~$f$ extends to an automorphism on~$M^n$.
\hfill$\Box$

We now return to the general situation of Definition~\ref{D-MacWExt}.
We first show that the RT-weight satisfies the extension property on $\lR{M^n}$.
The following is the main step in doing so.

\begin{thm}\label{T-RTLT}
Let $n\in\N$ and let $\cC\subseteq\lR{M^n}$ be a code and $f: \cC \longrightarrow M^n$ a linear $\wtRT$-preserving map.
Then  for all $v\in \cC$ there exists a matrix $A_v\in\LT_n(R)$ such that $f(v) = vA_v$.
\end{thm}

\begin{proof}
Let $v=(v_1,\ldots ,v_n)\in\cC$ and set $f(v)=u=(u_1,\ldots , u_n)$.
Fix $j\in\{1,\ldots,n\}$.
Then $\cV_j := v_jR + \cdots +v_nR$ is a submodule of $\rR{M}$.
This allows us to make use of the double annihilator property.
Recall the notation from~\eqref{e-Perps}.
Note that if $r\in\lp{\cV_j}$ then $\wtRT(rv) < j$.
As~$f$ is a linear $\wtRT$-preserving map we conclude
\[
   \wtRT(ru) = \wtRT(rf(v)) = \wtRT(f(rv)) = \wtRT(rv) < j.
\]
This implies in particular that $ru_j = 0$ and all of this shows that $\lp{\cV_j}\subseteq \lp{(u_jR)}$.
Thus Corollary~\ref{C-KMR} yields
\[
   u_jR=\rp{(\lp{(u_jR)})}\subseteq\rp{(\lp{\cV_j})}=\cV_j.
\]
As a consequence, $u_j = \sum_{i=j}^n v_ia_{ij}$ for some $a_{ij}\in R$.
Setting $a_{ij}=0$ for $i<j$ we obtain a lower triangular matrix $A_v=(a_{ij})$ such that
$u=f(v) = vA_v$.
We know that any linear $\wtRT$-preserving map is injective and thus $f$ induces a $\wtRT$-isometry between $\cC$ and $\cC':=f(\cC)$.
Using the same argument as above for $f^{-1}$, we
observe that~(i) and~(ii) of Lemma~\ref{L-Local} are satisfied, and the lemma thus tells
us that we may assume that $A_v$ is invertible for each~$v\in\cC$.
\end{proof}

The matrices produced in Theorem~\ref{T-RTLT} are \emph{local matrices} as they depend on~$v$.
In the next step we collect this local information and create a global matrix that describes the given RT-isometry.
This will allow us to extend the isometry to a map on the entire module $M^n$.

\begin{thm}\label{T-RTExt}
Let $n\in\N$ and $\cC\subseteq\lR{M^n}$ be a code and $f: \cC \longrightarrow M^n$ be an RT-isometry.
Then there exists a matrix $A\in\LT_n(R)$ such that $f(v)=vA$ for all $v\in\cC$ and thus~$f$ extends to the
RT-isometry $v\longmapsto vA$ on $M^n$.
\end{thm}

\begin{proof}
By Theorem~\ref{T-RTLT}, for all $v\in \cC$ there exists $A_v \in \LT_n(R)$ such that $f(v) = vA_v$.
Let $\{e_1, \ldots, e_n\}$ be the standard basis for $\lR{R^n}$.
Then Lemma~\ref{L-AiInner} implies that for each $i = 1, \ldots, n$ there exists $A_i \in \LT_n(R)$ such that
$\ideal{f(v), e_i} =  vA_ie_i^\top$ for all $v\in\cC$. Define the matrix
$A = (A_1e_1^\top, \ldots, A_ne_n^\top)$.
Then~$A$ is clearly lower triangular and invertible because each~$A_i$ is.
Furthermore, by construction $f(v) = vA$ for all $v\in \cC$.
Clearly the map~$f$ extends to the RT-isometry $v\longmapsto vA$  on~$M^n$.
\end{proof}

In the next part we focus on the extension theorem for the Hamming weight and some other weights.
In order to cover several cases together we formulate the following result.
Recall from Definition~\ref{D-DualPartMod} that for a partition of~$M^n$ the $\chi$-left dual is a partition of~$R^n$.
As before, we denote by~$e_i$ the standard basis vectors of~$R^n$.
\begin{thm}\label{T-PartExt}
Let~$\cP$ be a reflexive partition of~$\lR{M^n}$ and let $\cC\subseteq\lR{M^n}$ be a code.
Suppose $f: \cC \longrightarrow M^n$ is a linear map that preserves the partition, that is,
\begin{equation}\label{e-vfv}
    v\widesim_{\cP}f(v)\ \text{ for all }v\in\cC.
\end{equation}
Assume further that $S\subseteq R\setminus\{0\}$ is a subset such that the set
$Q:=\{s e_i\mid s\in S,\,i=1,\ldots,n\}$ is a block of the dual partition~$\wcPchil$.
Then
\begin{alphalist}
\item If $S$ is a subgroup of~$R^*$ there exists a matrix $A\in\Mon_{S,n}(R)$ such that $f(v)=vA$ for all $v\in\cC$.
\item If $R^*S:=\{\alpha s\mid \alpha\in R^*,\,s\in S\}=S$ and $1\in S$ there exists a matrix $A\in\Mon_n(R)$ such that
$f(v)=vA$ for all $v\in\cC$.
\end{alphalist}
\end{thm}
Note that the requirement in~(b) implies that $R^*\subseteq S$.

\begin{proof}
The group $\Hom(\lR{\cC},\lR{M})$ of left $R$-linear maps is a right $R$-module via $(gr)(v)=g(v)r$ for all $r\in R$ and $v\in\cC$.
Denote by $f_1,\ldots,f_n$ the coordinate functions of~$f$ and by $\pi_1,\ldots,\pi_n$ the projections of~$M^n$ onto its components.
Then $f_i,\,\pi_i\in\Hom(\lR{\cC},\lR{M})$.

We have to show that there exists a permutation~$\tau\in S_n$ and $s_1,\ldots,s_n\in R^*$ such that
$f_i=\pi_{\tau(i)}s_i$, and where $s_i\in S$ if $S$ is a subgroup of~$R^*$.

Since~$Q$ is a block of~$\wcPchil$ and~$\cP$ is reflexive,~\eqref{e-vfv} along with Theorem~\ref{T-chiBiDual} tells us that
\[
   \sum_{y\in Q}\chi(\ideal{f(v),y})=\sum_{y\in Q}\chi(\ideal{v,y})\text{ for all }v\in\cC.
\]
Hence
\begin{equation}\label{e-charSum}
   \sum_{y\in Q}\chi(\ideal{f(\sbt),y})=\sum_{y\in Q}\chi(\ideal{\sbt,y})
\end{equation}
as sums of characters on~$\cC$.
By Remark~\ref{R-CharGroup}(6) the character $\chi(\ideal{f(\sbt),e_1})$ must appear on the right hand side of~\eqref{e-charSum}.
Hence there exists some~$s_1\in S$ and some $\tau(1)\in\{1,\ldots,n\}$ such that $\chi(\ideal{f(\sbt),e_1})=\chi(\ideal{\sbt,s_1 e_{\tau(1)}})$.
Now Theorem~\ref{T-FrobBimod} implies that the linear maps $\ideal{f(\sbt),e_1}$ and $\ideal{\sbt, s_1e_{\tau(1)}}$ coincide on~$\cC$.
This means $f_1=\pi_{\tau(1)}s_1$.

It remains to show for this step that we may choose $s_1$ in~$R^*$.
In the case~(a) where $S$ is a subgroup of~$R^*$, this is clearly the case for $s_1$.
For the case~(b) we argue as follows.
Consider the submodules $f_1R,\ldots,f_nR,\pi_1R,\ldots,\pi_nR$ of the right module $\Hom(\lR{\cC},\lR{M})$.
Without loss of generality we may assume that $f_1R$ is maximal among all these submodules, that is, it is not properly contained
in any of these submodules (the situation is indeed symmetric in~$f_i$ and~$\pi_i$ because the latter are the coordinate
functions of the identity map on~$\cC$).
Now the identity $f_1=\pi_{\tau(1)}s_1$ with $s_1\in S$ along with maximality of~$f_1$ implies  $f_1R=\pi_{\tau(1)}R$ and
Theorem~\ref{T-LeftSim} yields the existence of some unit $u_1\in R^*$ such that $f_1=\pi_{\tau(1)}u_1$.
Hence in both cases we arrived at
\begin{equation}\label{e-f1pi}
    f_1=\pi_{\tau(1)}s_1,\ \text{ where $s_1\in R^*$, and where $s_1\in S$ if~$S$ is a subgroup of~$R^*$.}
\end{equation}
Now we return to~\eqref{e-charSum}.
The set~$Q$ is the disjoint union of the subsets $Q_i=\{s e_i\mid\in S\},\,i=1,\ldots,n$.
Moreover,
\begin{align*}
  \sum_{y\in Q_1}\chi(\ideal{f(\sbt),y})&=\sum_{s\in S}\chi(\ideal{f(\sbt),s e_1})=\sum_{s\in S}\chi(\ideal{f(\sbt),e_1}s)
        =\sum_{s\in S}\chi(\ideal{\sbt,s_1 e_{\tau(1)}}s)\\
        &=\sum_{s\in S}\chi(\ideal{\sbt,s_1s e_{\tau(1)}})=\sum_{y\in Q_{\tau(1)}}\chi(\ideal{\sbt,y}),
\end{align*}
where the very last step follows from the identity $S=s_1 S$, which is true in both cases for the set~$S$ and by the choice of~$s_1$ as in~\eqref{e-f1pi}.

As a consequence, the identity~\eqref{e-charSum} may be reduced to
\begin{equation}\label{e-charSum2}
   \sum_{y\in Q\setminus Q_1}\chi(\ideal{f(\sbt),y})=\sum_{y\in Q\setminus Q_{\tau(1)}}\chi(\ideal{\sbt,y}).
\end{equation}
Now we repeat the argument with the character $\chi(\ideal{f(\sbt),e_2})$ appearing on the left hand side of~\eqref{e-charSum2}.
It must also appear on the right hand side and as above we obtain
\[
    f_2=\pi_{\tau(2)}s_2,\ \text{ where $s_2\in R^*$, and $s_2\in S$ if $S$ is a subgroup of~$R^*$.}
\]
Clearly, $\tau(2)\neq\tau(1)$.
Proceeding in this manner we arrive at the desired result.
\end{proof}

Now we obtain the extension property for the Hamming weight.
This has also been established earlier by Greferath et al.~\cite[Thm.~4.15]{GNW04}.
\begin{thm}\label{T-MacWEHamm}
The Hamming weight on~$M^n$ satisfies the MacWilliams extension property for all $n\in\N$.
In other words, for every code $\cC\subseteq \lR{M^n}$ and every Hamming-weight preserving linear map $\lR{\cC}\longrightarrow\lR{M^n}$ there exists a
monomial matrix $A\in\Mon_n(R)$ such that $f(v)=vA$ for all $v\in\cC$.
\end{thm}
\begin{proof}
The Hamming partition~$\cPHam$ on~$M^n$ is reflexive and its left $\chi$-dual is again the Hamming partition; see Example~\ref{E-HammPart}.
Consider the set $S=R\setminus\{0\}$.
Then the set~$Q$ of Theorem~\ref{T-PartExt} is exactly the set of vectors of Hamming weight~$1$ and thus a block of~$\widehat{\cPHam}^{[\chi,l]}$.
Moreover,~$S$ satisfies Condition~(b) of Theorem~\ref{T-PartExt} and, since every Hamming-weight preserving map clearly preserves $\cPHam$ in
the sense of~\eqref{e-vfv}, the result follows.
\end{proof}

We now turn to symmetrized weight compositions in the following sense.
Recall that $\Aut(\lR{M})\cong R^*$ in the natural way; see~\eqref{e-EndMn}.
\begin{defi}
Let $G\leq R^*$ be a subgroup of~$R^*$ and denote by $M/G$ the orbit space of the action $v\longmapsto vu$.
For a vector $v= (v_1,\ldots ,v_n) \in M^n$ and $m\in M/G$ define
\[
   \swc_m(v) = |\{i : v_i \in m\}|.
\]
The \emph{symmetrized weight composition (with respect to~$G$)} of a vector $v \in M^n$ is defined as
$\swc_G(v):=\big(\swc_m(v)\big)_{m\in M/G}$.
It encodes the number of entries of~$v$ that are contained in each of the orbits.
\end{defi}

Observe that $\swc_G$ is a refinement of the Hamming weight: if $\swc_G(v)=\swc_G(w)$ for some $v,w\in M^n$, then $\wtH(v)=\wtH(w)$, and
thus the partition induced by $\swc_G$ is finer than the Hamming weight partition.
As a special case $\swc_{\{1\}}(v)$ is the \emph{complete weight} of~$v$, i.e., it counts the number of entries equal to a
given module element; see \cite[p.~142]{MS77}.

We consider now linear maps that preserve the symmetrized weight composition.
Using Theorem~\ref{T-PartExt} we can establish the extension theorem for such maps.
It implies in particular that these maps are given by monomial matrices.
The following result has been proven in a different manner by ElGarem et al.~\cite[Thm.~13]{EMW15} and also appears
in \cite[Thm.~8.1]{Wo09} by Wood, where it was established with the aid of averaging characters.

\begin{thm}\label{P-swcMn}
Let $\cC\subseteq\lR{M^n}$ and $f:\lR{\cC} \longrightarrow\lR{M^n}$ be a linear map and let $G\leq R^*$ be a subgroup.
Suppose~$f$ preserves the symmetrized weight composition, i.e.,  $\swc_G\big(f(v)\big)=\swc_G(v)$ for all $v\in\cC$.
Then there exists a matrix $A\in\Mon_{G,n}(R)$ such that $f(v)=vA$ for all $v\in\cC$.
As a consequence, the $\swc_G$-preserving maps on~$M^n$ are given by $G$-monomial maps.
\end{thm}
\begin{proof}
Let $U=\Mon_{G,n}(R)\leq\GL_n(R)$ and consider the partition $\cP=\cP_{M^n,U}$ on~$M^n$.
The assumption on~$f$ implies that $f(v)\widesim_{\mathcal P}v$ for all $v\in \cC$.
Furthermore, by Theorem~\ref{T-UPart} this partition is reflexive and $\wcPchil=\cP_{R^n,U^{\sf T}}$.
This shows that $Q:=\{\alpha e_i\mid i=1,\ldots,n,\,\alpha\in G\}$ is a block of $\wcPchil$.
Now Theorem~\ref{T-PartExt}(a) concludes the proof.
\end{proof}

\begin{rem}\label{R-UPreserving}
In this context the reader may wonder about subgroups $\cU\subseteq\GL_n(R)$ and its orbits in $M^n$.
What can be said about linear maps preserving these orbits?
Suppose $\cC\subseteq\lR{M^n}$ and $f:\cC\longrightarrow M^n$ is a linear map such that
\[
   v\widesim_{\cP_{M^n,\cU}}f(v)\ \text{ for all }v\in\cC.
\]
This translates into the property: for all $v\in\cC$ there exists a matrix $U_v\in\cU$ such that $f(v)=vU_v$.
However, the latter does not necessarily imply that there exists a global matrix $U\in\cU$ such that $f(v)=vU$ for all $v\in\cC$.
An example for the case where~$M$ is a finite field is given in \cite[Ex.~2.6]{BGL15}.
In~\cite{BGL15} the authors discuss which subgroups of~$\GL_n(R)$ satisfy such a `local-global property' for maps on~$R^n$, where~$R$
is a Frobenius ring.
See also Remark~\ref{R-OtherWeights}(b) for a special case where the extension property does hold true.
\end{rem}

The methods employed above may also be used to prove the extension theorem for the homogeneous weight on~$M^n$;
see Definition~\ref{D-homogWt} and~\eqref{e-AddWeights}.
The following result appears also in~\cite[Thm.~4.15]{GNW04} by Greferath et al.
We provide an alternative, and much shorter, proof which makes use of the character-based formula for the homogeneous weight derived
in Theorem~\ref{T-HomogChar}.
Recall from Example~\ref{E-HomogRingMod}(a) that the partition induced by the homogeneous weight is in general not reflexive.
Therefore, a direct application of Theorem~\ref{T-PartExt} is not possible.
But, as we will see, we can derive an identity that allows us to continue in the same manner as in the proof of that theorem.

\begin{thm}\label{T-HomogWeightExt}
Let $\cC\subseteq\lR{M^n}$ and $f:\lR{\cC} \longrightarrow\lR{M^n}$ an $\omega$-preserving map, where~$\omega$ is the
homogeneous weight on~$M$ extended additively to~$M^n$.
Then there exists some matrix $A\in\Mon_n(R)$ such that $f(v)=vA$ for all $v\in\cC$.
In particular,~$f$ extends to a $\omega$-preserving map on~$M^n$.
As a consequence, a linear map $f:\cC \longrightarrow M^n$  preserves the homogeneous weight iff it preserves the Hamming weight.
\end{thm}
\begin{proof}
First of all, by the very definition of the homogeneous weight on~$M^n$ it is clear that for every $A\in\Mon_n(R)$ the map
$v\longmapsto vA$ is $\omega$-preserving.
This shows the last two statements of the theorem.
Next, define $Q=\{\alpha e_i\mid \alpha\in R^*,\,i=1,\ldots,n\}$.
From Theorem~\ref{T-HomogChar} we obtain
$\omega(v_1,\ldots,v_n)=n-\frac{1}{|R^*|}\sum_{i=1}^n\sum_{\alpha\in R^*}\chi(v_i\alpha)=
    n-\frac{1}{|R^*|}\sum_{y\in Q}\chi(\ideal{v,y})$.
Thus $f$ is $\omega$-preserving iff
\[
  \sum_{y\in Q}\chi(\ideal{f(\sbt),y})=\sum_{y\in Q}\chi(\ideal{\sbt,y}).
\]
But the latter is exactly~\eqref{e-charSum} for the subgroup $S=R^*$.
Proceeding as in the proof of Theorem~\ref{T-PartExt} we arrive at the desired matrix $A\in\Mon_n(R)$ such that $f(v)=vA$ for all $v\in\cC$.
\end{proof}

One should also note that by the above result, any linear $\omega$-preserving map $f:\cC \longrightarrow M^n$ is injective.
This is not a priori clear because nonzero vectors in~$M^n$ may have zero homogeneous weight, see Example~\ref{E-HomogRingMod}(a),
and thus could be in the kernel of~$f$.

Finally, we discuss a weight function that does not satisfy in general the extension property.
\begin{rem}\label{R-WtN}
Let $N\subseteq\lR{M}$ be a submodule and for $v\in M^n$  define the $N$-weight
\[
   \wtN(v)=|\{i\mid v_i\in N\}|.
\]
Thus, $\wtN$ counts the number of entries in the submodule~$N$. For which~$N$ does $\wtN$ satisfy the extension property?
We have two extreme cases.
Firstly, if~$N=\{0\}$, then $\wtN(v)=n-\wtH(v)$, and thus~$\wtN$ satisfies the extension property by Theorem~\ref{T-MacWEHamm}.
Next, if $N=M$, then $\wtN(v)=n$ for all $v\in M^n$ and the extension property simply requires that any linear map
$\cC\longrightarrow M^n$ extends to a linear map on~$M^n$.
This is indeed the case by injectivity of the module $M^n$.
However, for nontrivial choices of~$N$ the weight does in general not satisfy the extension property.
This is most easily seen by choosing $\cC= N^n$.
Then of course the zero map $\cC\longrightarrow R^2$ is $\wtN$-preserving and does not extend as such.
But even an injective $\wtN$-preserving map does not extend in general.
We provide an example.
Let $M=R:=\Z_{24}$ and $N=(6)=\{0,6,12,18\}$. Consider the map
\[
  f:R^2\longrightarrow R^2,\ (x,y)\longmapsto (x,y)\begin{pmatrix}2&1\\3&1\end{pmatrix}.
\]
Since the rightmost matrix is in $\GL_2(R)$ this is even an isomorphism.
The map is not $\wtN$-preserving because, for instance, $f(3,0)=(6,3)$.
Choose $\cC=\{(0,y)\mid y\in (2)\}$.
Then $f|_{\cC}:\cC\longrightarrow R^2$ is clearly $\wtN$-preserving and injective.
Suppose now that $\hat{f}: R^2\longrightarrow R^2$ is a $\wtN$-preserving extension of $f|_{\cC}$.
Then by linearity there exists a matrix $A:=\Smallfourmat{a}{b}{c}{d}\in \cM_2(R)$ such that $\hat{f}(x,y)=(x,y)A$.
Using $\hat{f}|_{\cC}=f|_{\cC}$, one concludes that $c\in \{3,15\}$, and $d\in \{1,13\}$.
But then one can easily verify that for each such choice of~$A$ the resulting map is not $\wtN$-preserving on~$R^2$.
\end{rem}

We close the section with the following overview of further cases where the extension property can be established with the
methods developed in this paper.
We refer to~\cite{BGL15} for details in the case where $M=R$ is a Frobenius ring.
\begin{rem}\label{R-OtherWeights}\
\begin{alphalist}
\item Let $f:\cC\longrightarrow M^n$ be a support-preserving linear map, that is, $\supp(x)=\supp(f(x))$ for all $x\in\cC$, and
      where $\supp(x)$ is defined in the obvious way.
      Then~$f$ extends to a support-preserving map on~$M^n$ and these maps are given by right multiplication with a diagonal matrix
      in~$\GL_n(R)$.
      This is shown similarly to Theorem~\ref{T-RTExt}, see also \cite[Thm.~6.3]{BGL15}.
\item For $i=1,\ldots,n$ let $G_i$ be a subgroup of~$R^*$ and let $\cP_i:=\cP_{M,G_i}$ be the orbit partition of~$M$ induced by~$G_i$ in the sense of
       Definition~\ref{D-UAction} (for $n=1$).
       Let $f:\cC\longrightarrow M^n$ be a linear map such that for any $x=(x_1,\ldots,x_n)\in\cC$ and
       $f(x)=(y_1,\ldots,y_n)$ we have $x_i\widesim_{\cP_i}y_i$ for $i=1,\ldots,n$.
       That is,~$f$ preserves the $G_i$-orbit in each coordinate.
       Then~$f$ extends to an automorphism on~$M^n$ with the same preserving property, and~$f$ is given by right multiplication with a
       diagonal matrix in $\GL_n(R)$ where the $i$th diagonal entry is from~$G_i$.
       This can be shown similar to the proof of Theorem~\ref{T-RTExt} along with the fact that~$f$ is support-preserving, see~(a),
       and Lemma~\ref{L-AiInner}.
\item Suppose $M^n=M^{n_1}\times\cdots\times M^{n_t}$ and each~$M^{n_i}$ is endowed with a weight function~$\wt_i$ that satisfies the
      extension property and such that $\wt_i(x)=0\Leftrightarrow x=0$ for all~$x\in M^{n_i}$.
      Let $f:\cC\longrightarrow M^n$ be a linear map such that $(\wt_1(x_1),\ldots,\wt_t(x_t))=(\wt_1(f_1(x)),\ldots,\wt_t(f_t(x)))$ for all
      $(x_1,\ldots,x_t)\in\cC$ and where $f_i$ is the $i$th coordinate map of~$f$.
      Then~$f$ extends to an automorphism on~$M^n$ with the same weight-list-preserving property.
      This follows straightforwardly as in \cite[Rem.~6.8]{BGL15}.
\item Finally, the notion of poset weights~\cite{BGL95,KiOh05,PFKH08} can be naturally extended to codes over module alphabets, and it follows that
      a poset weight satisfies the extension property if and only if the poset is hierarchical.
      This can be established -- with some straightforward modifications -- as in \cite[Sec.~7]{BGL15}, where it was shown for Frobenius ring alphabets.
      This result generalizes both Theorem~\ref{T-RTExt} and Theorem~\ref{T-MacWEHamm}.
      See also the next section for a similar, but much simpler, situation.
\end{alphalist}
\end{rem}

\section{The RT-Weight for Sublinear Maps}
In \cite[Thm.~5.2]{Wo09} Wood proved that the Hamming weight satisfies the extension property on a module alphabet if and only if the module
is pseudo-injective and has a cyclic socle.
This means in particular that the extension property fails in the simple case where the alphabet is a proper field extension
$\F_{q^r}$ over a field $\F_q$ (thus $r>1$) because, clearly, the socle of the $\F$-vector space $\F_{q^r}$ is not cyclic.
A nice simple example of a non-extendable $\F_q$-linear $\wtH$-preserving map was found by Dyshko~\cite[Ex.~5]{Dys15}.
(However, in the same paper the author showed that the extension property does hold for codes in $\F_{q^r}^n$ if $n<q$.)

In this section we will show that, different from the Hamming weight, the RT-weight does satisfy the extension property for
$\F_q$-subspaces~$\cC$ of $\F_{q^r}^n$
and $\F_q$-linear maps $f:\cC\longrightarrow \F_{q^r}^n$ and any length~$n$.
Precisely, we consider a field extension $\hat{\F}:=\F_{q^r}$ of~$\F:=\F_q$ and
endow $\hat{\F}^n$ with the RT-weight from Definition~\ref{D-RTWeight}.
Moreover, we consider $\F$-linear maps on $\hat{\F}^n$ or $\F$-linear subspaces thereof and call such maps \textit{sublinear}.
Note that~$\hat{\F}$ is not an $(\F,\F)$-Frobenius bimodule.
However, the methods derived earlier can be utilized to prove the following result.
\begin{thm}\label{T-RTSublinear}
Let $\cC\subseteq\hat{\F}^n$ be an $\F$-subspace of $\hat{\F}^n$ and let $f:\cC\longrightarrow\hat{\F}^n$ be an $\F$-linear map
that preserves the RT-weight on $\hat{\F}^n$.
Then~$f$ extends to an RT-weight-preserving $\F$-linear map on~$\hat{\F}^n$.
\end{thm}

We need some preparation.
Let $\phi:\hat{\F}\longrightarrow\F^r$ be any $\F$-isomorphism.
Then $\hat{\F}^n$ is isomorphic as an~$\F$-vector space to $(\F^r)^n=\F^{rn}$ via
$(x_1,\ldots,x_n)\longmapsto(\phi(x_1),\ldots,\phi(x_n))$, and we call this isomorphism again~$\phi$.
We denote the RT-weight (see Definition~\ref{D-RTWeight}) on~$(\F^r)^n$ by
$\wtRTr(v_1,\ldots,v_n)$.
Then clearly $\wtRTr(v_1,\ldots,v_n)=\wtRT(\phi^{-1}(v_1),\ldots,\phi^{-1}(v_n))$ for all $v_i\in\F^r$.
Moreover, for every linear map $f:\cC\longrightarrow\hat{\F}^n$, where $\cC\subseteq\hat{\F}^n$, we have that $f$ is $\wtRT$-preserving iff
$f':=\phi\circ f\circ\phi^{-1}:\phi(\cC)\longrightarrow\F^{rn}$ is $\wtRTr$-preserving.
Define $R$ as the matrix ring $R=\cM_r(\F)$.
We will make use of the ring $\LT_n(R)$ defined in~\eqref{e-LTR}, which in this case is thus the ring of all
invertible lower triangular block matrices (which thus have matrices from $\GL_r(\F)$ on their diagonal).

Now we can prove the following preliminary result.

\begin{prop}\label{P-RTSublinear}
Let $f:\hat{\F}^n\longrightarrow\hat{\F}^n$ be an $\F$-linear map
that preserves the RT-weight on $\hat{\F}^n$.
Define $f':=\phi\circ f\circ\phi^{-1}:\F^{rn}\longrightarrow\F^{rn}$.
Then $f'$ is $\wtRTr$-preserving and there exists a matrix $A\in\LT_n(R)$ such that $f(v)=vA$ for all $v\in\F^{rn}$.
\end{prop}

\begin{proof}
Clearly $f'$ is $\wtRTr$-preserving.
For the existence of a matrix~$A$  we follow the proof of Proposition~\ref{P-Isom}.
Consider the restrictions of~$f'$ to the subspaces
$\tau_i(\F^{rn})=\{(v_1,\ldots,v_i,0,\ldots,0)\mid v_j\in \F^r\}$.
Then~$f'$ induces $\wtRTr$-preserving automorphisms of~$\tau_i(\F^{rn})$.
Proceeding consecutively for $i=1,\ldots,n$ and using that the isomorphisms of~$\F^r$
are given by $\GL_r(\F)$ yields the desired result.
\end{proof}

Note that the above result does not mean that the map~$f:\hat{\F}^n\longrightarrow\hat{\F}^n$ is given by
right multiplication with a lower triangular matrix with entries in~$\hat{\F}$.
This is simply due to the fact that  $\GL_2(\F_2)\not\cong\F_4^*$.

\medskip
Now we are ready to prove Theorem~\ref{T-RTSublinear}.

\noindent {\it Proof of Theorem~\ref{T-RTSublinear}}.
Set $\cC':=\phi(\cC)$ and $f':=\phi\circ f\circ\phi^{-1}:\cC'\longrightarrow\F^{rn}$.
Then $f'$ is clearly~$\F$-linear and $\wtRTr$-preserving.
It suffices to show that~$f'$ extends to a $\wtRTr$-preserving linear map on~$\F^{rn}$.
We induct on~$n$.
Let~$n=1$. In this case the $\wtRTr$-preserving property simply means injectivity.
Since the map~$f'$ can clearly be extended to an isomorphism on~$\F^r$ the case $n=1$ is established.

Consider now the general case.
Set $\widehat{\cC} = \{(v_1,\ldots,v_{n-1})\mid (v_1, \ldots v_{n-1},0) \in \cC'\}\subseteq\F^{r(n-1)}$.
Note that if $f'(v_1,\ldots , v_{n-1},v_n) =(w_1,\ldots ,w_{n-1},w_n)$ then $v_n = 0$ if and only if $w_n = 0$.
This implies that the map $\widehat{f} : \widehat{\cC} \longrightarrow\F^{r(n-1)},\ (v_1,\ldots,v_{n-1})\longmapsto(w_1,\ldots,w_{n-1})$,
where  $(w_1,\ldots,w_{n-1})$ is such that $f'(v_1, \ldots v_{n-1},0)=(w_1, \ldots w_{n-1},0)$, is a $\wtRTr$-isometry.
By induction there exists a matrix
\[
   A'=\begin{pmatrix}
   A_{11} & 0 & \cdots & 0 \\
A_{21} & A_{22} & \cdots & 0 \\
\vdots & \ddots & \ddots & \vdots \\
A_{n-1,1} & \cdots & A_{n-1,n-2} & A_{n-1,n-1}
  \end{pmatrix}\in\LT_{n-1}(R)
\]
such that $\widehat{f}(v)=vA'$ for all $v\in\widehat{\cC}$.
This yields
\begin{equation}\label{e-Indn1}
   f'(v_1,\ldots,v_{n-1},0)=((v_1,\ldots,v_{n-1})A',0)\ \text{ for all }(v_1,\ldots,v_{n-1},0)\in\cC'.
\end{equation}
Let~$\pi_n$ be the projection on the $n$-th coordinate.
Then the map $f'_n :\pi_n(\cC') \longrightarrow \F^r$ given by $v_n \longmapsto\pi_n(f'(v_1,\ldots ,v_n))$,
where $(v_1,\ldots,v_n)$ is any vector in~$\cC'$ with last component~$v_n$, is well defined and
an isometry.
The base case provides us with a matrix $A_{n,n} \in\GL_r(\F)$ such that $f'_n(v_n) = v_nA_{n,n}$.

Our next step is to find for every $v\in\cC'$ a matrix $A_v\in\LT_n(R)$ such that $f'(v)=vA_v$.
In order to do so, fix $v=(v_1,\ldots,v_n)\in\cC'$ such that $v_n\neq0$ and set $f'(v)=(w_1,\ldots,w_n)$.
Then we can clearly find matrices $A_{n,j,v}\in\cM_r(\F)$ such that
\[
   w_j-\sum_{i=j}^{n-1} v_i A_{ij}=v_n A_{n,j,v}\ \text{ for }j=1,\ldots,n-1.
\]
Setting $A_{n,j,0}=0\in\cM_r(\F)$ we obtain matrices
\[
    A_v=\begin{pmatrix}
      A_{11} & 0 & \cdots & 0 \\
      \vdots & \ddots &\ddots  &\vdots \\
      A_{n-1,1} & \cdots  & A_{n-1,n-1}&0\\
      A_{n,1,v}&\cdots&A_{n,n-1,v}&A_{nn}
  \end{pmatrix}\in\LT_{n}(R)\text{ for all }v\in\cC
\]
which, by construction, satisfy $f'(v)=vA_v$ for all $v\in\cC'$.

It remains to find a global matrix~$A\in\LT_n(R)$ such that $f'(v)=vA$ for all $v\in\cC'$.
This is indeed possible because the matrices on the diagonal do not depend on~$v$.
Set $A'':=\text{diag}(A_{11},\ldots,A_{nn})\in\GL_n(R)$ and consider the map
$\tilde{f}:\cC'\longrightarrow\F^{rn},\ v\longmapsto vA_v(A'')^{-1}$, which is clearly $\wtRTr$-preserving.
It suffices to show that $\tilde{f}$ extends to a $\wtRTr$-preserving map on~$\F^{rn}$.
Since the diagonal blocks of $A_v(A'')^{-1}$ are identity matrices, this matrix is actually in $\LT_{rn}(\F)$ and
thus~$\tilde{f}$ is even $\wtRT$-preserving on~$\cC\subseteq\F^{rn}$.
By Theorem~\ref{T-RTExt} the map extends to a $\wtRT$-preserving linear map on~$\F^{rn}$.
Since such a map is clearly also $\wtRTr$-preserving, the proof is complete.
\mbox{}\hfill$\Box$

The same procedure as above may also be used to establish the extension property for certain poset weights and sublinear maps.
This results in generalizations of \cite[Thms.~7.4 and 7.6]{BGL15} to sublinear maps.

\bibliographystyle{abbrv}

\end{document}